\documentclass[journal,draftcls,onecolumn,12pt,twoside]{IEEEtranTCOM}
\normalsize
\usepackage[latin1]{inputenc}
\usepackage{amsfonts}
\usepackage{amssymb}
\usepackage{amsmath}
\usepackage{graphicx}
\usepackage{amsmath}
\usepackage{float}
\usepackage{graphicx}
\usepackage{color}%
\providecommand{\U}[1]{\protect\rule{.1in}{.1in}}

\newtheorem{theorem}{Theorem}
\newtheorem{corollary}{Corollary}

\newtheorem{proposition}{Proposition}

\newenvironment{proof}{\noindent{\bf Proof:} }{\hfill $\Box$ \newline}
\providecommand{\U}[1]{\protect\rule{.1in}{.1in}}

\makeindex

\begin{document}

\title{Coding and Decoding Schemes for MSE and Image Transmission}

\author{Marcelo~Firer,~Luciano~Panek~and~Jerry~Anderson~Pinheiro
\thanks{M. Firer is with IMECC-UNICAMP, University of Campinas,
Brazil (e-mail: mfirer@ime.unicamp.br).}
\thanks{L. Panek is with CECE-UNIOESTE, Universidade Estadual
do Oeste do Paraná, Brazil (e-mail: lucpanek@gmail.com).}
\thanks{J. A. Pinheiro is with
IMECC-UNICAMP, University of Campinas, Brazil (e-mail:
jerryapinheiro@gmail.com).}}

\markboth{IEEE Transactions on Information Theory}%
{Submitted paper}

\maketitle
\IEEEpeerreviewmaketitle

\begin{abstract}
In this work we explore possibilities for coding and decoding tailor-made for
mean squared error evaluation of error in contexts such as image transmission.
To do so, we introduce a \emph{loss function} that expresses the overall
performance of a coding and decoding scheme for discrete channels and that exchanges the usual
goal of minimizing the error probability to that of minimizing the expected
loss. In this environment we explore the possibilities of using ordered decoders
to create a message-wise unequal error protection (UEP), where the most valuable
information is protected by placing in its proximity information words that
differ by a small valued error. We give explicit examples, using
scale-of-gray images, including small-scale performance analysis and visual
simulations for the BSMC.
\end{abstract}


\section{Introduction\label{Intro}}

\IEEEPARstart{I}{nformation} of various nature (pictures, movies, voice, music, text, etc.) is
compressed using different methods and algorithms (JPEG, MPEG, FLAC, MP3, PDF,
etc.), that takes into account the nature of the information and the different
loss in distortion that may be caused by different errors. When it comes to
information transmission, information is generally considered as just 
sequences of bits and the actual type of the information is generally ignored. In
this work, we consider instances where we do know something about the semantic
value of decoding errors. The main idea is the trade-off between the quantity and
importance of decoding errors, as is usually done in distortion theory.

We start by defining a general \emph{expected loss function}, the valued measure
of decoding errors, in Section \ref{back}. With this broad definition in
mind, we establish various existence results showing the importance of encoders
and decoders, in Section \ref{relevance}.

We give a first heuristic approach for image transmission in Section \ref{ap2},
proposing strategies to address the problem of transmitting images over a
(very) noisy channel. For this purpose, we consider the mean squared error
(MSE) as the value function of decoding errors, proposing both coding and decoding schemes and
considering a (syndrome) decoding algorithm that has the advantage of 
very low complexity. This heuristic approach is developed for a small
dimensional case, and includes performance analyses and visual test performed
on the sample in Figure \ref{mico}.

\begin{figure}[htbp!]
    \centering
       \includegraphics[scale=0.25]{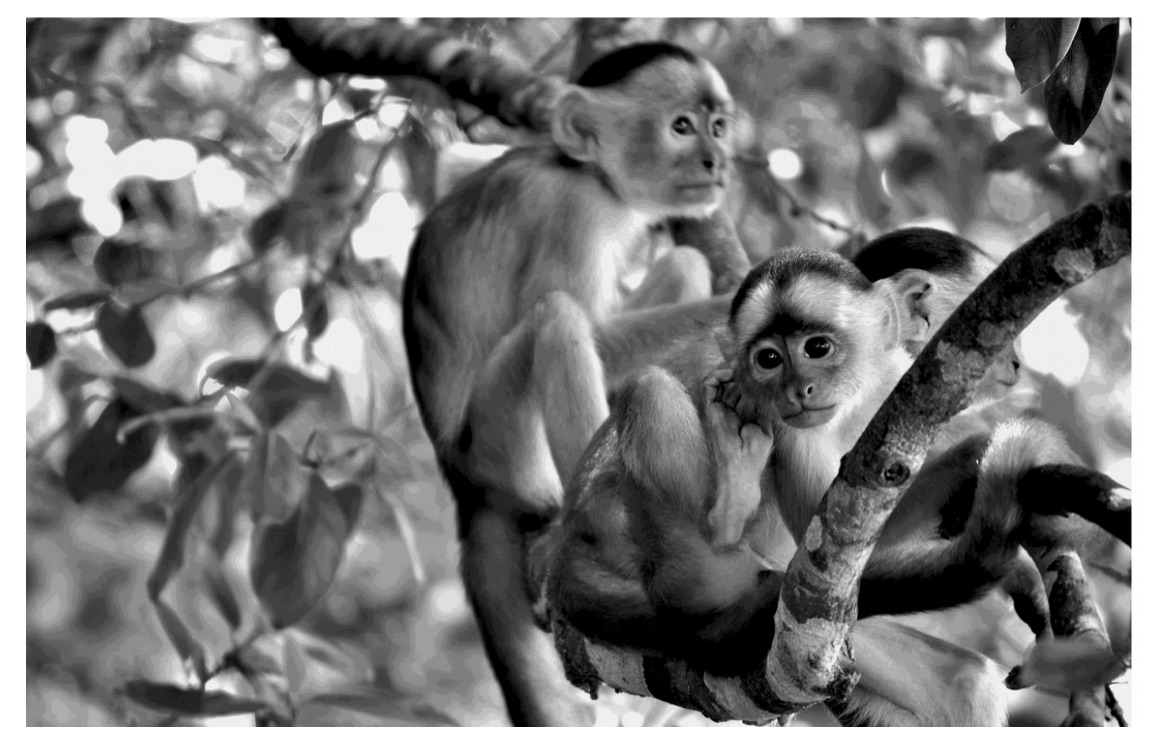}
    \label{mico}
   \caption{A family of tamarins at the Rio Negro, Amazon rain forest.}
\end{figure}

Finally, in Section \ref{translation}, we briefly explore the expected loss function when the error measure of decoding is invariant by translations, a situation where the BER (\textit{bit error probability}) is a particular case of this.

\section{Expected Loss Function\label{back}}

\label{section 1}

In this section we define the expected loss function of a code. We begin by
introducing some basic concepts and notation, following \cite{Ash},
\cite{Bigliere} and \cite{Cover}.

In this work, we consider a \textit{discrete channel} with \textit{input set}
$\mathcal{X}$ and \textit{output set} $\mathcal{Y}$, and we assume that
$\mathcal{X}\subseteq\mathcal{Y}$. We assume that both $\mathcal{X}$ and
$\mathcal{Y}$ are finite sets. The channel is determined by the matrix
$\left(  P\left(  y|x\right)  \right)  _{y\in\mathcal{Y},x\in\mathcal{X}}$ of
conditional probabilities $P\left(  y\text{ is received}|x\text{ is
sent}\right)  $. In subsequent sections we will consider $\mathcal{X}%
=\mathcal{Y}=\mathbb{F}_{q}^{n}$, the $n$-dimensional vector space over the
finite field $\mathbb{F}_{q}$ with $q$ elements, so we stress that $P\left(
y|x\right)  $ is considered over the words $x=\left(  x_{1},\ldots
,x_{n}\right)$ and $y=\left(  y_{1},\ldots,y_{n}\right)  $. We denote a discrete
channel by $P_{\mathcal{X},\mathcal{Y}}$.

We consider an \emph{information set }$\mathcal{I}{\ }$with $M=\left\vert
\mathcal{I}\right\vert $ elements. A \emph{coding-decoding scheme} for $\mathcal{I}$ over $P_{\mathcal{X},\mathcal{Y}}$ is a
triple {$\left(  C,f,a\right)  $ consisting of: $(i)$ a \emph{code}
$C\subseteq\mathcal{X}$ having }$M$ elements{; $(ii)$ a \emph{channel encoder
}$f:\mathcal{I}\rightarrow C$, where }$f$ may be any bijective map; and {$(iii)$ a
}\emph{channel decoder }consisting of a surjective map {$a:\mathcal{Y}%
\rightarrow C$}. In such a coding-decoding scheme, a piece of information $\iota\in\mathcal{I}$ is
encoded as a codeword $c=f\left(  \iota\right)  \in C$, and, after $c$ is transmitted,  a
message $y\in\mathcal{Y}$ is received and decoded as $a\left(  y\right)  $.
Because the encoder $f$ is assumed to be a bijection over $C$, the decoding
process is completed by assuming the original information was $f^{-1}\left(
a\left(  y\right)  \right)  \in\mathcal{I}$.

Let $D\left(  c\right)  $ be the \emph{decision region}\textit{ of} $c$
relative to the decoder $a$:
\[
D\left(  c\right)  :=a^{-1}\left(  c\right)  =\left\{  y\in\mathcal{Y}%
:a\left(  y\right)  =c\right\}  \text{.}%
\]
The decision regions $D\left(  c\right)  $ of a decoder $a$ determine a
partition of $\mathcal{Y}$. Given a coding-decoding scheme, an error occurs if $c$ is sent
and the received codeword lies in some decision region $D\left(  c^{\prime
}\right)  $ with $c^{\prime}\neq c$. The \textit{word error probability of}
$c$ is therefore
\[
P_{e}\left(  c\right)  =1-\sum_{y\in D\left(  c\right)  }P\left(  \left.
y\right\vert c\right)  \text{.}%
\]
Assuming that the probability distribution of $C$ is \textit{uniform}, that
is, each codeword $c$ is transmitted with probability $P\left(  c\right)
=\frac{1}{M}$ with $M:=\left\vert C\right\vert $, the \textit{word error
probability of }$C$ (WEP) is the average
\[
P_{e}\left(  C\right)  =\frac{1}{M}\sum_{c\in C}P_{e}\left(  c\right)
\text{.}%
\]

We  now let $\mathbb{R}_{+}$ denotes the set of non-negative real numbers and
consider the map
\[
\nu_{0\text{-}1}:C\times C\rightarrow\mathbb{R}_{+}%
\]
given by
\[
\nu_{0\text{-}1}\left(  c,c^{\prime}\right)  =\left\{
\begin{tabular}
[c]{ccc}%
$0$ & if & $c=c^{\prime}$\\
$1$ & if & $c\neq c^{\prime}$%
\end{tabular}
\ \right.  \text{.}%
\]
It follows that the word error probability may be expressed as
\begin{equation}
\label{probabilidade de erro}P_{e}\left(  C\right)  =\frac{1}{M}\sum_{c\in
C}\sum_{y\in\mathcal{Y}}\nu_{0\text{-}1}\left(  a\left(  y\right)  ,c\right)
P\left(  \left.  y\right\vert c\right)  \text{.}%
\end{equation}

The function $\nu_{0\text{-}1}$ is the \textit{indicator function}, which  
detects decoding errors (see \cite{Cover}) but does not distinguish
between different such errors.

In many real-world situations, such as the transmission of digital images, it is
reasonable to attribute different values to different errors, and this is what
will be done in this work, replacing $\nu_{0\text{-}1}$ by value
functions that may assume any (non-negative) real value.

An \emph{error value function }for a set of information \textit{$\mathcal{I}$
} is a map\textbf{\ }$\nu$ that associates to each pair of information a
non-negative real number
\[
\nu:\mathcal{I}\times\mathcal{I}\rightarrow\mathbb{R}_{+}%
\]
where $\nu\left(  \iota_{1},\iota_{2}\right)  $ is the cost of exchanging
$\iota_{2}$ by $\iota_{1}$. If $C$ is a code and $f:\mathcal{I}\rightarrow C$
is an encoder, we denote by
\[
\nu_{f}:C\times C\rightarrow\mathbb{R}_{+}%
\]
the \emph{error value function induced by the encoder $f$}: given $\iota
_{1},\iota_{2}\in\mathcal{I}$, we define
\[
\nu_{f}\left(  f\left(  \iota_{1}\right)  ,f\left(  \iota_{2}\right)  \right)
:=\nu\left(  \iota_{1},\iota_{2}\right)  .
\]

We shall refer to $\nu$ and $\nu_{f}$  as just \textit{error value functions}.
By considering such an error value function, we are interested in evaluating
the errors that may occur during the process consisting of coding,
transmitting and decoding information.

Given a coding-decoding scheme $\left(  C,f,a\right)  $ and a error value function $\nu
_{f}:C\times C\rightarrow\mathbb{R}_{+}$, let us denote by $\mathbb{E}%
(C,\left.  \nu_{f},a\right\vert y)$ the \emph{posterior expected loss }of the
induced error value function $\nu_{f}\left(  c,a\left(  y\right)  \right)  $,
given that $y$ is observed:
\begin{equation}
\mathbb{E}(C,\left.  \nu_{f},a\right\vert y)=\sum_{c\in C}\nu_{f}\left(
c,a\left(  y\right)  \right)  P\left(  \left.  c\right\vert y\right)  \text{.}%
\end{equation}
We define the\textit{\ }\emph{overall expected loss}\textit{ }of the coding-decoding scheme
as the average of the expected loss for all possible $y\in\mathcal{Y}$,
\[
\mathbb{E}\left(  C,a,\nu_{f}\right)  =\sum_{y\in\mathcal{Y}}\mathbb{E}%
(C,\left.  \nu_{f}\left(  c,a(y)\right)  \right\vert y)P(y),
\]
where $P(y)=\sum_{c\in C}P(c)P\left(  \left.  c\right\vert y\right)  $ is the
probability of receiving $y$. Considering Bayes' Rule, this expression can be
rewritten as
\[
\mathbb{E}\left(  C,a,\nu_{f}\right)  =\sum_{c\in C}\sum_{y\in\mathcal{Y}}%
\nu_{f}\left(  c,a\left(  y\right)  \right)  P\left(  \left.  y\right\vert
c\right)  P(c)
\]
and thus
\begin{equation}
\mathbb{E}\left(  C,a,\nu_{f}\right)  =\sum_{\left(  c,c^{\prime}\right)  \in
C\times C}G_{a}\left(  c,c^{\prime}\right)  \nu_{f}\left(  c,c^{\prime
}\right)  \label{perda esperada}%
\end{equation}
where
\begin{equation}
G_{a}\left(  c,c^{\prime}\right)  =\sum_{y\in a^{-1}\left(  c\right)
}P\left(  \left.  c^{\prime}\right\vert y\right)  P\left(  y\right)  \text{.}
\label{polinomio G}%
\end{equation}

In a general setting, we consider the following data to be given:

\begin{itemize}
\item The error value function $\nu$ (determined by the nature of the information);

\item The size of the code $C$ (determined by $\left\vert \mathcal{I}%
\right\vert $);

\item The rate $\left\vert \mathcal{I}\right\vert /\left\vert \mathcal{Y}%
\right\vert =\left\vert C\right\vert /\left\vert \mathcal{Y}\right\vert $
(determined by cost constraints);

\item The channel model $P_{\mathcal{X},\mathcal{Y}}$ (determined by physical conditions).
\end{itemize}

In such a setting, we say that the triple $\left(  C^{\ast},f^{\ast},a^{\ast
}\right)  $ is a \emph{Bayes coding-decoding scheme }if
\[
\mathbb{E}\left(  C^{\ast},a^{\ast},\nu_{f^{\ast}}\right)  =\min_{\left(
C,a,f\right)  }\mathbb{E}\left(  C,a,\nu_{f}\right)
\]
where the minimum is taken over all encoding-decoding schemes for
$\mathcal{I}$ over $P_{\mathcal{X},\mathcal{Y}}$.

As expected, determining a Bayes coding-decoding scheme is a (very) hard problem, so we may
consider each of the variables $C$, $f$ and $a$ independently and say that
$a^{\ast}$ is a \emph{Bayes decoder} of the pair $\left(  C_{0},f_{0}\right)
$ if
\[
\mathbb{E}\left(  C_{0},a^{\ast},\nu_{f_{0}}\right)  =\min_{\left(
C_{0},a,f_{0}\right)  }\mathbb{E}\left(  C_{0},a,\nu_{f_{0}}\right)  \text{.}%
\]
Similarly, we say that $f^{\ast}$ is a \emph{Bayes encoder} of the pair
$\left(  C_{0},a_{0}\right)  $ if
\[
\mathbb{E}\left(  C_{0},a_{0},\nu_{f^{\ast}}\right)  =\min_{\left(
C_{0},a_{0},f\right)  }\mathbb{E}\left(  C_{0},a_{0},\nu_{f}\right)  \text{.}%
\]

\section{Relevance of Encoders and Decoders\label{relevance}}

In this section, we show that, in quite general instances, every encoder and
every decoder may be relevant, depending on the error value functions to be considered.

We start by remarking that there are two classes of decoders that may be considered to
be \textquotedblleft universal," in the sense that they define a decoding criteria
$a:\mathcal{Y}\rightarrow C$ for any code $C$ in $\mathcal{X}$. The first
class is the probabilistic criterion determined by the channel, known as a
\textit{maximal likelihood decoder} (ML-decoder): given $y\in\mathcal{Y}$,
$a\left(  y\right)  \in C$ satisfies the inequality
\[
P\left(  y|a\left(  y\right)  \right)  \geq P\left(  y|c\right)  \text{ for
all }c\in C.
\]
The second class of universal decoders, relatives to the discrete channels
such that $\mathcal{X}=\mathcal{Y}$, are the so-called \textit{nearest
neighbor decoders} (NN-decoders) determined by a metric $d:\mathcal{X}%
\times\mathcal{X}\rightarrow\mathbb{R}_{+}$: given $y\in\mathcal{X}$,
$a\left(  y\right)  \in C$ satisfies the inequality
\[
d\left(  y,a\left(  y\right)  \right)  \leq d\left(  y,c\right)  \text{ for
all }c\in C.
\]

We start by proving that, for any linear code and any ML-decoder, there are always
error value functions for which it is better to use a different decoder.

From here on in this section, we assume that the prior probability of $C$ is
uniform.

{\color{red} }

\begin{theorem}
\label{principal1} Let $\left(  C,f,a\right)  $ be a coding-decoding scheme over a
reasonable discrete channel $P_{\mathcal{X},\mathcal{Y}}$, that is,
$P(x|x)>P(y|x)$ for all $y\neq x$. If $a:\mathcal{Y}\rightarrow C$ is an
ML-decoder, then there exists a decoder $b:\mathcal{Y}\rightarrow C$ and error
value functions $\nu_{f}$ and $\widetilde{\nu}_{f}$ such that
\[
\mathbb{E}\left(  C,a,\nu_{f}\right)  >\mathbb{E}\left(  C,b,\nu_{f}\right)
\text{ and }
\mathbb{E}\left(  C,a,\widetilde{\nu}_{f}\right)  <\mathbb{E}\left(
C,b,\widetilde{\nu}_{f}\right)  \text{.}%
\]

\end{theorem}

{}

\begin{proof}
Let $C$ be a code, $a:\mathcal{Y}\rightarrow C$ an
ML-decoder and $\nu_{f}:C\times C\rightarrow\mathbb{R}_{+}$ the error value
function defined by
\[
\nu_{f}\left(  c,c^{\prime}\right)  =\left\{
\begin{array}
[c]{ccc}%
1 & \text{ if } & c=c^{\prime}\\
0 & \text{ if } & c\neq c^{\prime}%
\end{array}
\right.  \text{.}%
\]
Note that $\nu_{f}$ does not depend of the encoder $f$ and that it may be expressed
as $\nu_{f}\left(  c,c^{\prime}\right)  =1-\nu_{0\text{-}1}\left(
c,c^{\prime}\right)  $.

We consider now $c_{1},c_{2}\in C$, with $c_{1}\neq c_{2}$, and define
$b:=b_{a,c_{1},c_{2}}:\mathcal{Y}\rightarrow C$ by
\[
b\left(  y\right)  =\left\{
\begin{array}
[c]{ccc}%
c_{1} & \text{ if } & y=c_{2}\\
c_{2} & \text{ if } & y=c_{1}\\
a\left(  y\right)  & \text{ otherwise } &
\end{array}
\right.  \text{.}%
\]

We are considering two different decoders and, for a given error value
function, we wish to look at the difference  $\mathbb{E}\left(  C,a,\nu
_{f}\right)  -\mathbb{E}\left(  C,b,\nu_{f}\right)  $. We define
\[
G_{a}\left(  c\right)     :=\sum_{y\in a^{-1}\left(  c\right)  }P\left(
c|y\right)  P\left(  y\right) \label{MM copy(1)}\text{ and }
G_{b}\left(  c\right)    :=\sum_{y\in b^{-1}\left(  c\right)  }P\left(
c|y\right)  P\left(  y\right)
\]
and thus, considering the error value function $\nu_{f}$, we may express the
difference%
\begin{equation}
\mathbb{E}\left(  C,a,\nu_{f}\right)  -\mathbb{E}\left(  C,b,\nu_{f}\right)
=\sum_{c\in C}G_{a}\left(  c\right)  -\sum_{c\in C}G_{b}\left(  c\right)
\label{mn}%
\end{equation}
because $\nu_{f}\left(  c,c^{\prime}\right)  =0$ if $c\neq c^{\prime}$ and
$\nu_{f}\left(  c,c\right)  =1$ for all $c\in C$.

Because for $y\neq c_{1},c_{2}$ we have that $a_{1}\left(  y\right)
=a_{2}\left(  y\right)  $, the equation (\ref{mn}) reduces to
\[
\mathbb{E}\left(  C,a,\nu_{f}\right)  -\mathbb{E}\left(  C,b,\nu_{f}\right)
=
\]
\[
P\left(  c_{1}|c_{1}\right)  P\left(  c_{1}\right)  +P\left(  c_{2}%
|c_{2}\right)  P\left(  c_{2}\right) 
-P\left(  c_{1}|c_{2}\right)  P\left(  c_{2}\right)  -P\left(
c_{2}|c_{1}\right)  P\left(  c_{1}\right)
\]
and assuming that the probability distribution of $C$ is uniform, we get
\[
\mathbb{E}\left(  C,a,\nu_{f}\right)  -\mathbb{E}\left(  C,b,\nu_{f}\right)
=
\]
\[
\frac{1}{M}\left(  P\left(  c_{1}|c_{1}\right)  +P\left(  c_{2}|c_{2}\right)
-P\left(  c_{2}|c_{1}\right)  -P\left(  c_{1}|c_{2}\right)  \right)  \text{.}%
\]
Because
\[
P\left(  c_{1}|c_{1}\right)  >P\left(  c_{1}|c_{2}\right)  \text{ and
}P\left(  c_{2}|c_{2}\right)  >P\left(  c_{2}|c_{1}\right),
\]
it follows that
\[
\mathbb{E}\left(  C,a,\nu_{f}\right)  -\mathbb{E}\left(  C,b,\nu_{f}\right)
>0\text{.}%
\]

To obtain the inequality
\[
\mathbb{E}\left(  C,a,\widetilde{\nu}\right)  -\mathbb{E}\left(
C,b,\widetilde{\nu}\right)  <0\text{,}%
\]
we can consider $\widetilde{\nu}$ to be  $\widetilde{\nu}=\nu_{0\text{-}1}%
$. Because $a$ is an ML-decoder, we have that
\[
\mathbb{E}\left(  C,a,\widetilde{\nu}\right)  < \mathbb{E}\left(
C,b,\widetilde{\nu}\right)  \text{.}%
\]

\end{proof}


\begin{corollary}
Suppose $\mathcal{X}=\mathcal{Y}$ and let $\left(  C,f,a\right)  $ be an
coding-decoding scheme over $\mathcal{X}$ such that $a:\mathcal{X}\rightarrow C$ is an
NN-decoder determined by a metric $d$. Then, there exists a discrete channel
$P_{\mathcal{X},\mathcal{X}}$, a decoder $b:\mathcal{X}\rightarrow C$ and
error value functions $\nu_{f}$ and $\widetilde{\nu}_{f}$ such that
\[
\mathbb{E}\left(  C,a,\nu_{f}\right)  >\mathbb{E}\left(  C,b,\nu_{f}\right)
\text{ and }
\mathbb{E}\left(  C,a,\widetilde{\nu}_{f}\right)  <\mathbb{E}\left(
C,b,\widetilde{\nu}_{f}\right)  .
\]

\end{corollary}

\begin{proof}
The proof follows from Theorem \ref{principal1} and from the fact that, given
a metric $d:\mathcal{X}\times\mathcal{X}\rightarrow\mathbb{R}_{+}$, there is a
discrete channel over $\mathcal{X}$ such that the NN-decoder determined by
$d$ and the ML-decoder coincide for any code $C\subseteq\mathcal{X}$ (proof to be found in \cite{WF}).
\end{proof}

From here on, we assume that $\mathcal{X}=\mathcal{Y}$ and write
$P_{\mathcal{X},\mathcal{Y}}=P_{\mathcal{X}}$.

\begin{theorem}
Let $a_{1}\neq a_{2}$ be two NN-decoders defined on $\mathcal{X}$, determined
respectively by the metrics $d_{1}$ and $d_{2}$. Then, there is a code
$C\subseteq\mathcal{X}$, encoders $f_{1},f_{2}:\mathcal{I}\rightarrow C$,
error value functions $\nu_{1},\nu_{2}:\mathcal{I}\times\mathcal{I}%
\rightarrow\mathbb{R}_{+}$ and an open family of discrete channels
$P_{\mathcal{X}}$ over $\mathcal{X}$ such that
\[
\mathbb{E}\left(  C,a_{1},\nu_{1}\right)  >\mathbb{E}\left(  C,a_{2},\nu
_{1}\right)
\text{ and }
\mathbb{E}\left(  C,a_{1},\nu_{2}\right)  <\mathbb{E}\left(  C,a_{2},\nu
_{2}\right)  \text{.}%
\]

\end{theorem}

\begin{proof}
Because we are assuming $a_{1}\neq a_{2}\,$, there is a code $C\subseteq
\mathcal{X}$ and $y_{0}\in\mathcal{X}\setminus C$ such that
\[
c_{1}:=a_{1}\left(  y_{0}\right)  \neq a_{2}\left(  y_{0}\right)  :=c_{2}%
\]
Because we are considering NN-decoders determined by metrics, say $d_{1}$ and
$d_{2}$, we have that
\[
d_{1}\left(  y_{0},c_{1}\right)     \leq d_{1}\left(  y_{0},c\right)  \text{ and }
d_{2}\left(  y_{0},c_{2}\right)    \leq d_{2}\left(  y_{0},c\right)  ,
\]
for all $c\in C$, and in particular
\[
d_{1}\left(  y_{0},c_{1}\right)     \leq d_{1}\left(  y_{0},c_{2}\right) \text{ and }
d_{2}\left(  y_{0},c_{2}\right)     \leq d_{2}\left(  y_{0},c_{1}\right)  ,
\]
hence we may assume, without loss of generality, that $C=\left\{  c_{1},c_{2}\right\} $

and thus our information set $\mathcal{I}=\left\{  \iota_{1},\iota_{2}\right\}  $
has only two elements. We consider the two possible encoders $f_{1}%
,f_{2}:\mathcal{I\rightarrow C}$ by%
\[
f_{1}\left(  \iota_{1}\right)  =c_{1},f_{1}\left(  \iota_{2}\right)  =c_{2}%
\text{ and }f_{2}\left(  \iota_{1}\right)  =c_{2},f_{2}\left(  \iota_{2}\right)
=c_{1}\text{.}%
\]

Let us consider the induced error value functions%
\[
\nu_{1}\left(  c,c^{\prime}\right)  =\left\{
\begin{array}
[c]{cc}%
\delta_{1} & \text{ if }\left(  c,c^{\prime}\right)  =\left(  c_{1}%
,c_{2}\right) \\
0 & \text{ otherwise}%
\end{array}
\right.
\]
and
\[
\nu_{2}\left(  c,c^{\prime}\right)  =\left\{
\begin{array}
[c]{cc}%
\delta_{2} & \text{ if }\left(  c,c^{\prime}\right)  =\left(  c_{2}%
,c_{1}\right) \\
0 & \text{ otherwise}%
\end{array}
\right.
\]

Let us define%
\[
V_{i}^{j}=\left\{  y\in\mathcal{X}:a_{j}\left(  y\right)  =c_{i}\right\}  ,
\]
$i,j\in\left\{  1,2\right\}  $. With this notation, for $i\neq j$, we have
that
\begin{align*}
\mathbb{E}\left(  C,a_{j},\nu_{i}\right)   &  =\sum_{y\in\mathcal{X}^{n}}%
\sum_{c\in C}\nu_{i}\left(  c,a_{j}\left(  y\right)  \right)  P\left(
y|c\right) \\
&  =\sum_{y\in\mathcal{X}^{n}}\nu_{i}\left(  c_{i},a_{j}\left(  y\right)
\right)  P\left(  y|c_{i}\right) \\
&  =\sum_{y\in V_{j}^{j}}\delta_{i}\cdot P\left(  y|c_{i}\right) 
\end{align*}
and
\begin{align*}
\mathbb{E}\left(  C,a_{i},\nu_{i}\right)   &  =\sum_{y\in\mathcal{X}^{n}}%
\sum_{c\in C}\nu_{i}\left(  c,a_{i}\left(  y\right)  \right)  P\left(
y|c\right) \\
&  =\sum_{y\in\mathcal{X}^{n}}\nu_{i}\left(  c_{i},a_{i}\left(  y\right)
\right)  P\left(  y|c_{i}\right) \\
&  =\sum_{y\in V_{j}^{i}}\delta_{i}\cdot P\left(  y|c_{i}\right)  \text{.}%
\end{align*}
Considering the difference  between the expected loss functions we find that%
\begin{align}
\mathbb{E}\left(  C,a_{1},\nu_{1}\right)  -\mathbb{E}\left(  C,a_{2},\nu
_{1}\right) & =
 \delta_{1}\left(  \sum_{y\in V_{2}^{1}}P\left(
y|c_{1}\right)  -\sum_{y\in V_{2}^{2}}P\left(  y|c_{1}\right)  \right) \\
 & =  \delta_{1}\left(  \sum_{y\in V_{2}^{1}\backslash V_{2}^{2}}P\left(
y|c_{1}\right)  -\sum_{y\in V_{2}^{2}\backslash V_{2}^{1}}P\left(
y|c_{1}\right)  \right)
\label{v1}
\end{align}
and, similarly,
\begin{align}
\mathbb{E}\left(  C,a_{1},\nu_{2}\right)  -\mathbb{E}\left(  C,a_{2},\nu
_{2}\right) =
\delta_{2}\left(  \sum_{y\in V_{1}^{1}\backslash V_{1}^{2}%
}P\left(  y|c_{2}\right)  -\sum_{y\in V_{1}^{2}\backslash V_{1}^{1}}P\left(
y|c_{2}\right)  \right)  \text{.} \label{v2}%
\end{align}

Considering $i\neq j$, we remark now that%
\begin{align*}
V_{j}^{i}\backslash V_{j}^{j}  &  =\left\{  y\in\mathcal{X}:a_{i}\left(
y\right)  =c_{j},a_{j}\left(  y\right)  \neq c_{j}\right\} \\
&  =\left\{  y\in\mathcal{X}:a_{i}\left(  y\right)  =c_{j},a_{j}\left(
y\right)  =c_{i}\right\} \\
&  =\left\{  y\in\mathcal{X}:a_{i}\left(  y\right)  \neq c_{i},a_{j}\left(
y\right)  =c_{i}\right\}  =V_{i}^{j}\backslash V_{i}^{i}%
\end{align*}
and%
\begin{align*}
V_{i}^{i}\backslash V_{i}^{j}  &  =\left\{  y\in\mathcal{X}:a_{i}\left(
y\right)  =c_{i},a_{j}\left(  y\right)  \neq c_{i}\right\} \\
&  =\left\{  y\in\mathcal{X}:a_{i}\left(  y\right)  =c_{i},a_{j}\left(
y\right)  =c_{j}\right\} \\
&  =\left\{  y\in\mathcal{X}:a_{i}\left(  y\right)  \neq c_{j},a_{j}\left(
y\right)  =c_{j}\right\}  =V_{j}^{j}\backslash V_{j}^{i},
\end{align*}
along equations (\ref{v1}) and (\ref{v2}), can be represented as%
\[
\mathbb{E}\left(  C,a_{1},\nu_{1}\right)  -\mathbb{E}\left(  C,a_{2},\nu
_{1}\right)  =
\delta_{1}\left(  \sum_{y\in V_{2}^{1}\backslash V_{2}^{2}%
}P\left(  y|c_{1}\right)  -\sum_{y\in V_{2}^{2}\backslash V_{2}^{1}}P\left(
y|c_{1}\right)  \right)  \label{v3}%
\]
and
\[
\mathbb{E}\left(  C,a_{1},\nu_{2}\right)  -\mathbb{E}\left(  C,a_{2},\nu
_{2}\right)  =
\delta_{2}\left(  \sum_{y\in V_{2}^{2}\backslash V_{2}^{1}%
}P\left(  y|c_{2}\right)  -\sum_{y\in V_{2}^{1}\backslash V_{2}^{2}}P\left(
y|c_{2}\right)  \right)  \text{.}%
\]

For simplicity, let us write
\[
V_{1}=V_{2}^{1}\backslash V_{2}^{2},V_{2}=V_{2}^{2}\backslash V_{2}%
^{1}\text{.}%
\]
Thus, we have that%
\[
\mathbb{E}\left(  C,a_{1},\nu_{1}\right)  -\mathbb{E}\left(  C,a_{2},\nu
_{1}\right)  =
\delta_{1}\left(  \sum_{y\in V_{1}}P\left(  y|c_{1}\right)
-\sum_{y\in V_{2}}P\left(  y|c_{1}\right)  \right)
\]
and
\[
\mathbb{E}\left(  C,a_{1},\nu_{2}\right)  -\mathbb{E}\left(  C,a_{2},\nu
_{2}\right)  =
\delta_{2}\left(  \sum_{y\in V_{2}}P\left(  y|c_{2}\right)
-\sum_{y\in V_{1}}P\left(  y|c_{2}\right)  \right)  \text{.}%
\]
Now, we note that%
\[
\mathbb{E}\left(  C,a_{1},\nu_{1}\right)  -\mathbb{E}\left(  C,a_{2},\nu
_{1}\right)  >0 \text{ and }
\mathbb{E}\left(  C,a_{1},\nu_{2}\right)
-\mathbb{E}\left(  C,a_{2},\nu_{2}\right)  <0
\]
is equivalent to having
\begin{equation}
\sum_{y\in V_{1}}P\left(  y|c_{1}\right)  >\sum_{y\in V_{2}}P\left(
y|c_{1}\right) \label{v6.1} \text{ and }
\sum_{y\in V_{1}}P\left(  y|c_{2}\right)
>\sum_{y\in V_{2}}P\left(  y|c_{2}\right)  . 
\end{equation}
Because $V_{1}\cap V_{2}=\varnothing$, there is a channel $P_{0}=P_{\mathcal{X}%
,\mathcal{X}}^{0}$, and there are$\ y\in\mathcal{X}$ and $c_{1},c_{2}\in C$,
satisfying the inequalities in (\ref{v6.1}). Because this is a strict inequality, it will also be satisfied
for any channel $P$ sufficiently close to $P_{0}$.
\end{proof}

\begin{corollary}
\label{cor}Let $C\subset\mathcal{X}$ be a code, $f:\mathcal{I}\rightarrow C$
an encoder and $a,b:\mathcal{X}\rightarrow C$ two distinct NN-decoders
determined respectively by metrics $d_{1}$ and $d_{2}$. Suppose there are
$c_{1},c_{2}\in C$ such that
\[
\sum_{y\in V_{1}}P\left(  y|c_{1}\right)  >\sum_{y\in V_{2}}P\left(
y|c_{1}\right) \text{ and }
\sum_{y\in V_{1}}P\left(  y|c_{2}\right)
>\sum_{y\in V_{2}}P\left(  y|c_{2}\right)  ,
\]
where
\[
V_{1}=\left\{  y\in\mathcal{X}:a\left(  y\right)  =c_{2},b\left(  y\right)
=c_{1}\right\}\text{ and }
V_{2}=\left\{  y\in\mathcal{X}:a\left(  y\right)  =c_{1},b\left(  y\right)
=c_{2}\right\}  \text{.}%
\]
Then, there are error value functions $\nu_{f}$ and $\widetilde{\nu}_{f}$ such
that
\[
\mathbb{E}\left(  C,a,\nu_{f}\right)  >\mathbb{E}\left(  C,b,\nu_{f}\right)
\text{ and }\mathbb{E}\left(  C,a,\widetilde{\nu}_{f}\right)  <\mathbb{E}\left(
C,b,\widetilde{\nu}_{f}\right)  .
\]

\end{corollary}

\bigskip

\begin{proof}
Follows from the proof of the preceding theorem.
\end{proof}

\section{Coding and Decoding Schemes for Images \label{ap2}}

For the purpose we are aiming at with this work, we consider a gray palette of
colours, using a scale of $k$ bits. This means that our information set is
$\mathcal{I}_{k}=\left\{  \iota_{0},\iota_{1},\ldots,\iota_{2^{k}-1}\right\}  $,
where $\iota_{r}$, $0\leq\frac{r}{2^{k}-1}\leq1$,  represents the brightness of a
pixel in a scale of gray with $2^{k}$ uniform levels. We let $\mathbb{F}_{2}$
be the finite field with two elements denoted by $0$ and $1$. Each $\iota_j$ may
be represented as a binary vector $x=\left(  x_{0},x_{1},\ldots,x_{k-1}%
\right)  $, with $x_{i}\in\mathbb{F}_{2}$, which represents a color with a
black intensity of%

\[
\frac{x_{0}+x_{1}2+x_{2}2^{2}+\ldots+x_{k-1}2^{k-1}}{2^{k}-1}\text{,}%
\]
where here the $x_{i}$'s are considered to assume the \textit{\textbf{real}} values $0$
or $1$, that is, to the color $\iota_{r}$ with a black intensity of $\frac{r}{2^{k}-1}  $, we associate $x$ such that $r=x_{0}+x_{1}2+x_{2}%
2^{2}+\ldots+x_{k-1}2^{k-1}$. Given such $x,y\in\mathbb{F}_{2}^{k}$, the
\textit{squared error loss value} $\mu_{2}$ is%
\[
\mu_{2}\left(  x,y\right)     =\left[  \frac{\sum_{i=0}^{k}\left(  x_{i}-y_{i}\right)  2^{i}}{2^{k}%
-1}\right]  ^{2}\text{.}%
\]

Let us consider a picture $X$ that has $M\times N$ pixels, let us say
$X=\left(  x^{mn}\right)  $, with $1\leq m\leq M$ and $1\leq n\leq N$, and
suppose that, after the coding-transmitting-decoding process, we get a picture
$Y=\left(  y^{mn}\right)  $. The \textit{mean squared error} (MSE) $\mu
_{\text{{\tiny MSE}}}\left(  X,Y\right)  $ is%

\[
\mu_{\text{{\tiny MSE}}}\left(  X,Y\right)   =\frac{1}{MN}\sum_{m=1}%
^{M}\sum_{n=1}^{N}\mu_{2}\left(  x^{mn},y^{mn}\right)
\]
\[
=\frac{1}{MN}\frac{1}{\left(  2^{k}-1\right)  ^{2}}\sum_{m=1}^{M}\sum
_{n=1}^{N}\left(  \sum_{i=0}^{k}\left(  x_{i}^{mn}-y_{i}^{mn}\right)
2^{i}\right)  ^{2}\text{.}%
\]

This is a very simple measure of images distortion, which is considered to be
an appropriate measure of the fidelity of images when errors are produced by a
Gaussian noise, the same type of noise that gives rise to a symmetric channel.
Thus, from here on, we assume that transmission is made over a \emph{binary
memoryless symmetric channel} (BMSC) and that $\mu=\mu_{\text{{\tiny MSE}}}$.

We proceed now to present our proposed coding-decoding scheme, considering transmission of
images over a BMSC with overall error probability $p$. We consider linear
block-codes, so that the information set $\mathcal{I}_{k}$ has $2^{k}%
=\left\vert \mathcal{I}_{k}\right\vert $ elements that are encoded considering a
linear code $C\subseteq\mathbb{F}_{2}^{n}$, for some $n\geq k$. This is
still an initial approach where each codeword represents a pixel but it fits
into the concept of \emph{ultra-small block codes} as explored in \cite{Ning},
suitable for situations with strong constraints on block length (see, for
example, the introduction in \cite{Moser}).

We present our heuristic proposal in two parts, considering first the encoding
and then the decoding.

\subsection{Encoding\label{encoding}}

Given an $\left[  n;k\right]  _{q}$\textit{ linear code }$C$, a $k$-dimensional linear space of $\mathbb{F}_{q}^{n}$, and
assuming a decoder $a:\mathbb{F}_{q}^{n}\rightarrow C$ is given, we are
concerned with the choice of an encoder $f:\mathcal{I}\rightarrow
C$. If we fix such an encoder $f:\mathcal{I}\rightarrow C$, we are actually distinguishing
between $f$ and $f\circ\sigma$, where $\sigma:\mathcal{I}\rightarrow
\mathcal{I}$ is any permutation of the information set. In this sense,
we may say we are making a joint source-channel coding (JSCC), in the same
sense adopted for instance, in \cite{Ho-Kahn}, where some quantized information
is more relevant than others.

Approaching the encoding problem, first of all, we consider the peculiar (in
the coding context) situation where \textit{\textbf{no}} redundancy is added to the
system, that is, we are considering a $\left[  k;k\right]  _{q}$ linear code
$C$. Under this circumstance, and considering that the channel is \textit{reasonable},
in the sense that, for any $c\in C=\mathbb{F}_{q}^{k}$, $P\left(  \left.
c\right\vert c\right)  >P\left(  \left.  c\right\vert y\right)  $ for
any $y\in C$, $y\neq c$, we have that the unique ML decoder is the trivial
decoder: $a\left(  y\right)  =y$ for any $y\in C$. We identify the information
set $\mathcal{I}$ with $C$, and so, an encoder is just a permutation
$\sigma$ of $C$. Considering such permutation, we have that
\[
\mu_{0\text{-}1}\left(  c,c^{\prime}\right)  =\mu_{0\text{-}1}\left(
\sigma\left(  c\right)  ,\sigma\left(  c^{\prime}\right)  \right)
\]
for all $c,c^{\prime} \in C$, and hence, the error probability does not depend on the encoder. If, instead
of $\mu_{0\text{-}1}$, we consider a value function $\mu$ such that
\[
\mu\left(  c,c^{\prime}\right)  \neq\mu\left(  c,c^{\prime\prime
}\right)
\]
for some distinct $c,c^{\prime},c^{\prime\prime} \in C$ (as is such the
square error loss value $\mu_{2}$) then, if we exchange $c^{\prime}$ by
$c^{\prime\prime}$, the expected loss is affected.

To put it shortly: even in the most trivial case that can be considered, using
a code with no redundancy and no error correction or detection, better results
may be attained when the semantic value of errors is taken into consideration.

As a toy example, let the information set $\mathcal{I}_4$  consist of
$16=2^{4}$ different gray-scale tones $\left\{  \iota_{0},\ldots,\iota
_{15}\right\}  $. We consider two different encoders, $f,g:\mathcal{I}_{4}%
\rightarrow\mathbb{F}_{2}^{4}$. Encoder $f$ is determined by a
reflex-and-prefix algorithm used for producing \textit{Gray encoders}: it has the
property $d_H\left(  f\left(  \iota_{j}\right)  ,f\left(  \iota_{j+1}\right)
\right)  =1$, where $d_H\left(  \cdot,\cdot\right)  \,\ $is the usual \textit{Hamming
metric}. Encoder $g$ is determined by the lexicographic order: if we write
$g\left(  \iota_{j}\right)  =\left(  x_{1},x_{2},x_{3},x_{4}\right)  $, with
$x_{i}=0,1$, we have that
\[
j-1=x_{1}2^{0}+x_{2}2^{1}+x_{3}2^{2}+x_{4}2^{3}\text{.}%
\]
We consider as an original message the tamarins' picture in Figure
\ref{mico}.

Using a random number generator, we simulate a BMSC with overall bit error
probability $p=0.2$, and we get the two different \textquotedblleft decoded \textquotedblright messages, shown in
Figures \ref{fig:gray} and \ref{fig:lexi}. Since the overall error
probability $p$ is very high, we find that each picture has approximately
$0.4234$ \% of the pixels having a wrong color, and both pictures are
poor in quality. Nevertheless, the result using a lexicographic encoder is
clearly perceived to be better.

\begin{figure}[!h]
\centering
\begin{minipage}[b]{0.47\linewidth}
\includegraphics[width=\linewidth]{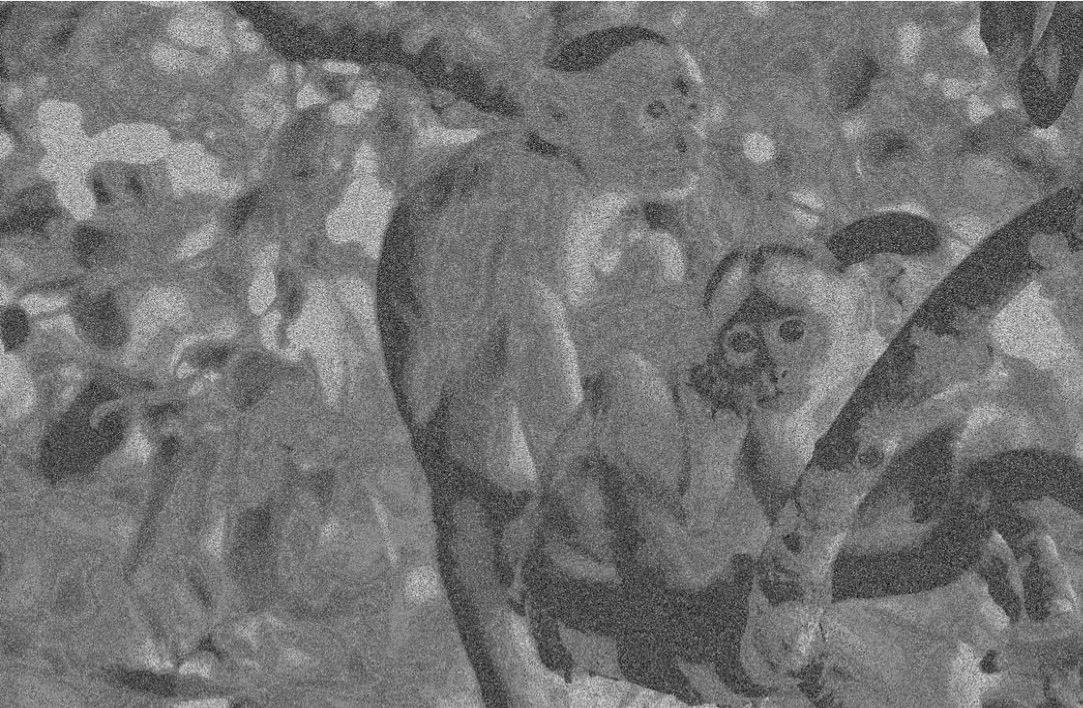}
\caption{Encoded using a reflex-and-prefix algorithm; $p=0.2$.}
\label{fig:gray}
\end{minipage} \hfill
\begin{minipage}[b]{0.47\linewidth}
\includegraphics[width=\linewidth]{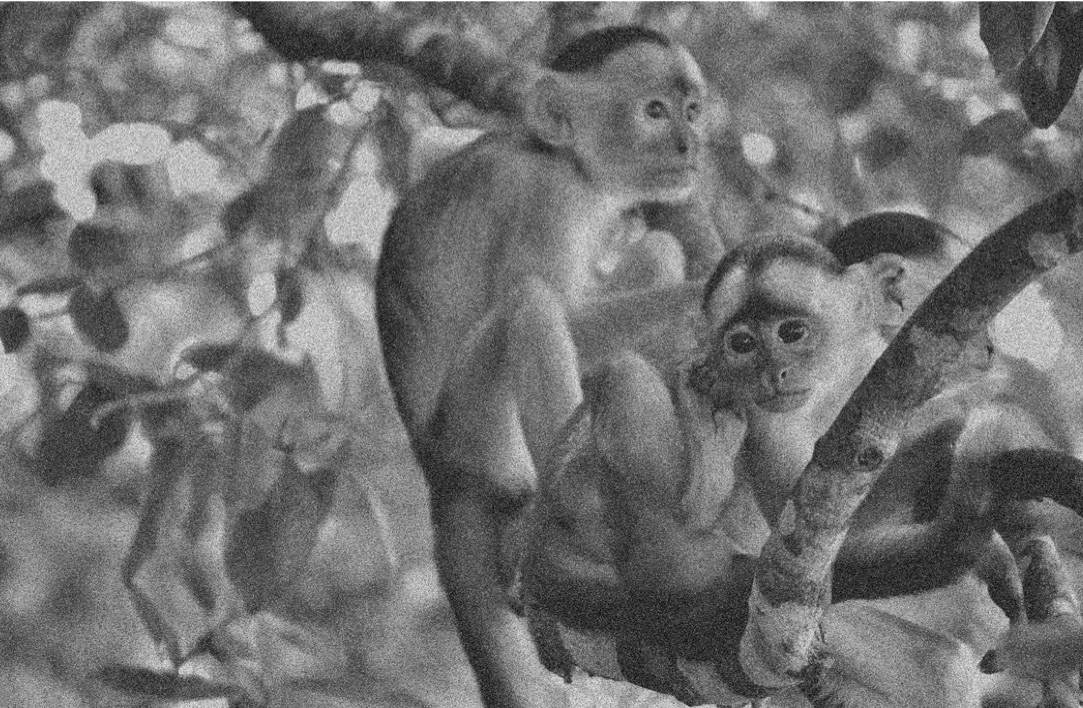}
\caption{Encoded using a lexicographic order; $p=0.2$.}
\label{fig:lexi}
\end{minipage}
\end{figure}

Simulations show that, in this situation, the lexicographic encoder seems to be
optimal. In Figure \ref{fig:lex_n=k=4}, we consider the
case $n=k=4$ and, in Figure \ref{fig:lex_n=k=8}, we consider the
case $n=k=8$. On both graphs, the red
lines represent the normalized expected loss attained by the lexicographic
encoder (as a function of the overall error probability $p$). The other values
correspond to different encoders, sampled randomly for different values of $p$.

\begin{figure}[!h]
\centering
\begin{minipage}[b]{0.47\linewidth}
\includegraphics[width=\linewidth]{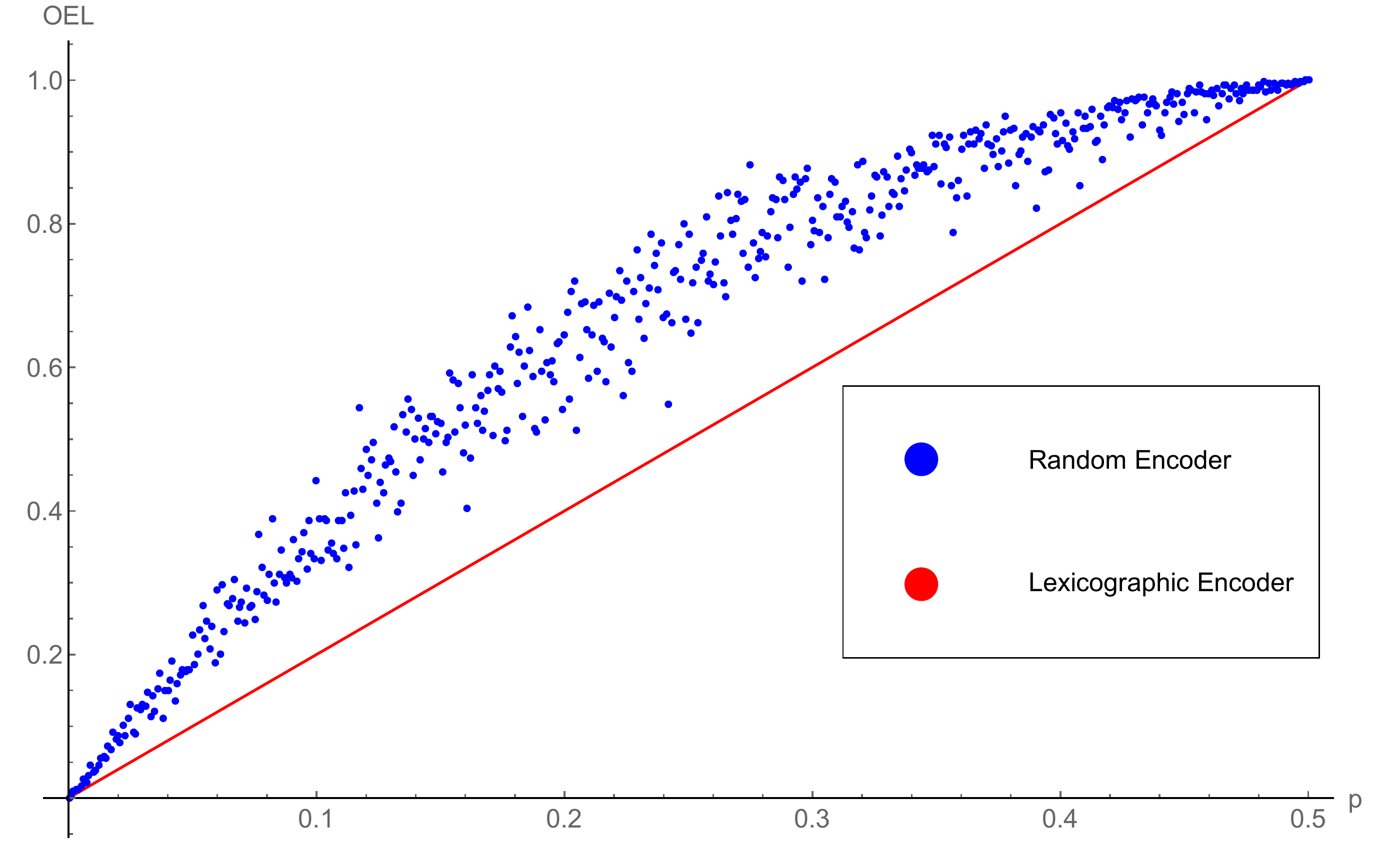}
\caption{$n=k=4$}
\label{fig:lex_n=k=4}
\end{minipage} \hfill
\begin{minipage}[b]{0.47\linewidth}
\includegraphics[width=\linewidth]{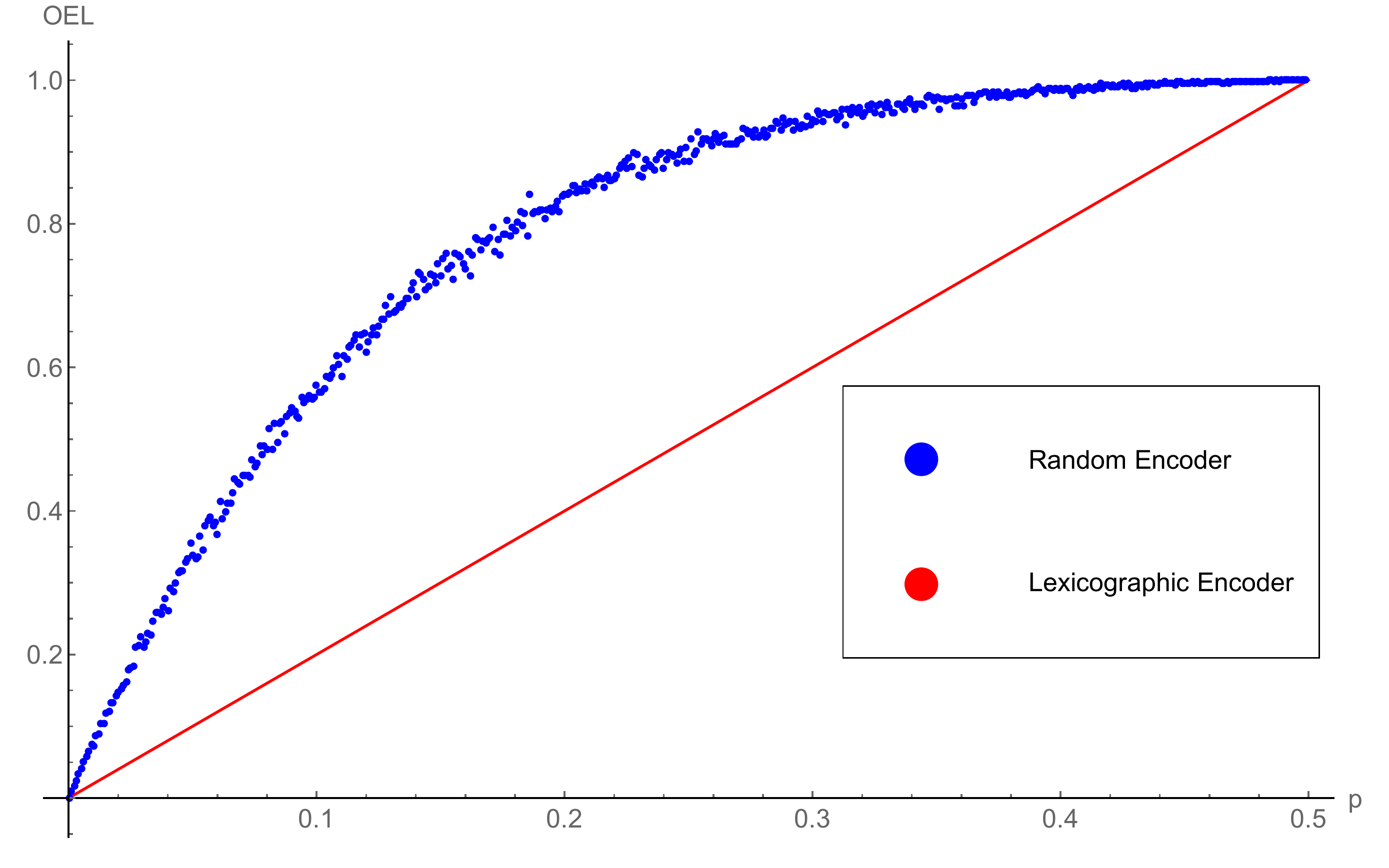}
\caption{$n=k=8$}
\label{fig:lex_n=k=8}
\end{minipage}
\end{figure}

We remark that, when $n=k$ increases, the difference in the performance
between the lexicographic encoder and a random encoder (in the average) also increases.

We generalize this approach when considering an information set $\mathcal{I}$
with $\left\vert \mathcal{I}\right\vert =2^{k}$ elements and a proper code
$C$ in $\mathbb{\mathbb{F}}_{2}^{n}$ with $n\geq k$. We say an encoder
$f:\mathcal{I}\rightarrow C$ is a \emph{lexicographic encoder} if it satisfies
the following condition: given $\iota_{j},\iota_{l}\in\mathcal{I}$, if
$f\left(  \iota_{j}\right)  =x=\left(  x_{1},\ldots,x_{n}\right)  $ and $f\left(
\iota_{l}\right)  =y=\left(  y_{1},\ldots,y_{n}\right)$, then
\[
j<l\iff\sum_{i=1}^{n}x_{i}2^{i-1}<\sum_{i=1}^{n}y_{i}2^{i-1}\text{.}%
\]

We cannot prove that this is indeed a Bayes encoder, but experimental
evidence supports this conjecture not only for the maximal likelihood (ML)
decoder for the BMSC but also for various other encoders. Let us consider the
situation described in our toy example where we have $2^{4}$ different tones
of gray and let us consider an $\left[  7;4\right]  _{2}$ linear code $C$. We
consider a decoder $a:\mathbb{F}_{2}^{7}\rightarrow C$ and compute the
normalized expected loss function for different encoders and different values
of $p$. On each of the pictures in Figures \ref{fig:HammingML}, \ref{fig:HammingT}, \ref{fig:EhammingML} and \ref{fig:EhammingT}, we have that:

\begin{itemize}
\item The red line represents the lexicographic encoder;

\item The black line represents an encoder  $f:\mathcal{I} \rightarrow C$ such
that, for $i\leq j$, we have $w_{H}\left(  f\left(  \iota_{i}\right)  \right)
\leq w_{H}\left(  f\left(  \iota_{j}\right)  \right)  $. Here, $w_{H}\left(
\cdot\right)  $ is the usual \textit{Hamming weight} of a vector. Such an encoder is
called a \emph{weight priority encoder} and, its role will be explained by Theorem \ref{teo. cod. bayes} in
Section \ref{translation} ;

\item The blue line represents an average (taken over all of the possible encoders)
expected loss;

\item The points represent random sampling of encoders for different values of
$p$.
\end{itemize}

We consider two different $\left[  7;4\right]  _{2}$ codes $\mathcal{H}\left(
3\right)  $ and $\mathcal{C}\left(  3\right)  $ and two different decoders, 
$a_{ML}$ and $a_{T}$, such that each of those four pictures represents a pair
$\left(  C,a\right)  =\left(  \text{code,decoder}\right)  $. The codes and
decoders are the following:

\begin{itemize}
\item $\mathcal{H}\left(  3\right)  $ is the Hamming code with a parity check
matrix
\[
\left(
\begin{array}
[c]{ccccccc}%
1 & 0 & 1 & 0 & 1 & 0 & 1\\
0 & 1 & 1 & 0 & 0 & 1 & 1\\
0 & 0 & 0 & 1 & 1 & 1 & 1
\end{array}
\right)  \text{;}%
\]

\item $\mathcal{C}\left(  3\right)  $ is the code with a generator matrix%
\[
\left(
\begin{array}
[c]{ccccccc}%
1 & 0 & 0 & 0 & 1 & 0 & 1\\
0 & 1 & 0 & 0 & 0 & 1 & 1\\
0 & 0 & 1 & 0 & 1 & 1 & 0\\
0 & 0 & 0 & 1 & 0 & 0 & 1
\end{array}
\right);
\]

\item $a_{ML}$ is the maximal likelihood decoder or the nearest (relative to
the Hamming distance $d_{H}$) neighbor decoder;

\item \label{ordem total}$a_{T}$ is the nearest neighbor decoder relatively to
the metric $d_{T}\left(  x,y\right)  :=\max\left\{  i:x_{i}\neq y_{i}\right\}
$, called the \emph{total-order decoder} or just \emph{ordered decoder} (more details about this decoder are explained in the next section).
\end{itemize}

\begin{figure}[!h]
\centering
\begin{minipage}[b]{0.47\linewidth}
\includegraphics[width=\linewidth]{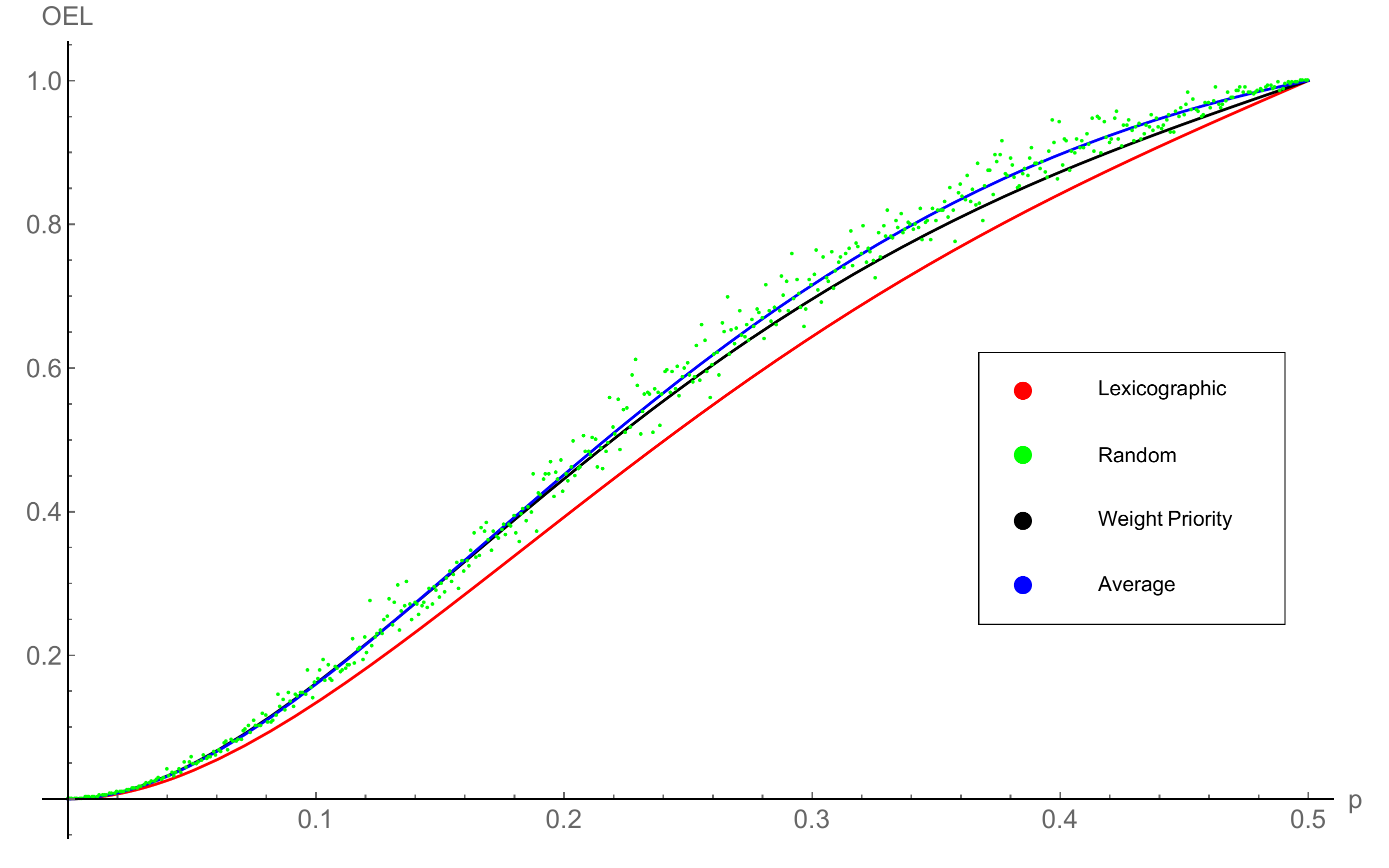}
\caption{$(\mathcal{H}(3),a_{ML})$}
\label{fig:HammingML}
\end{minipage} \hfill
\begin{minipage}[b]{0.47\linewidth}
\includegraphics[width=\linewidth]{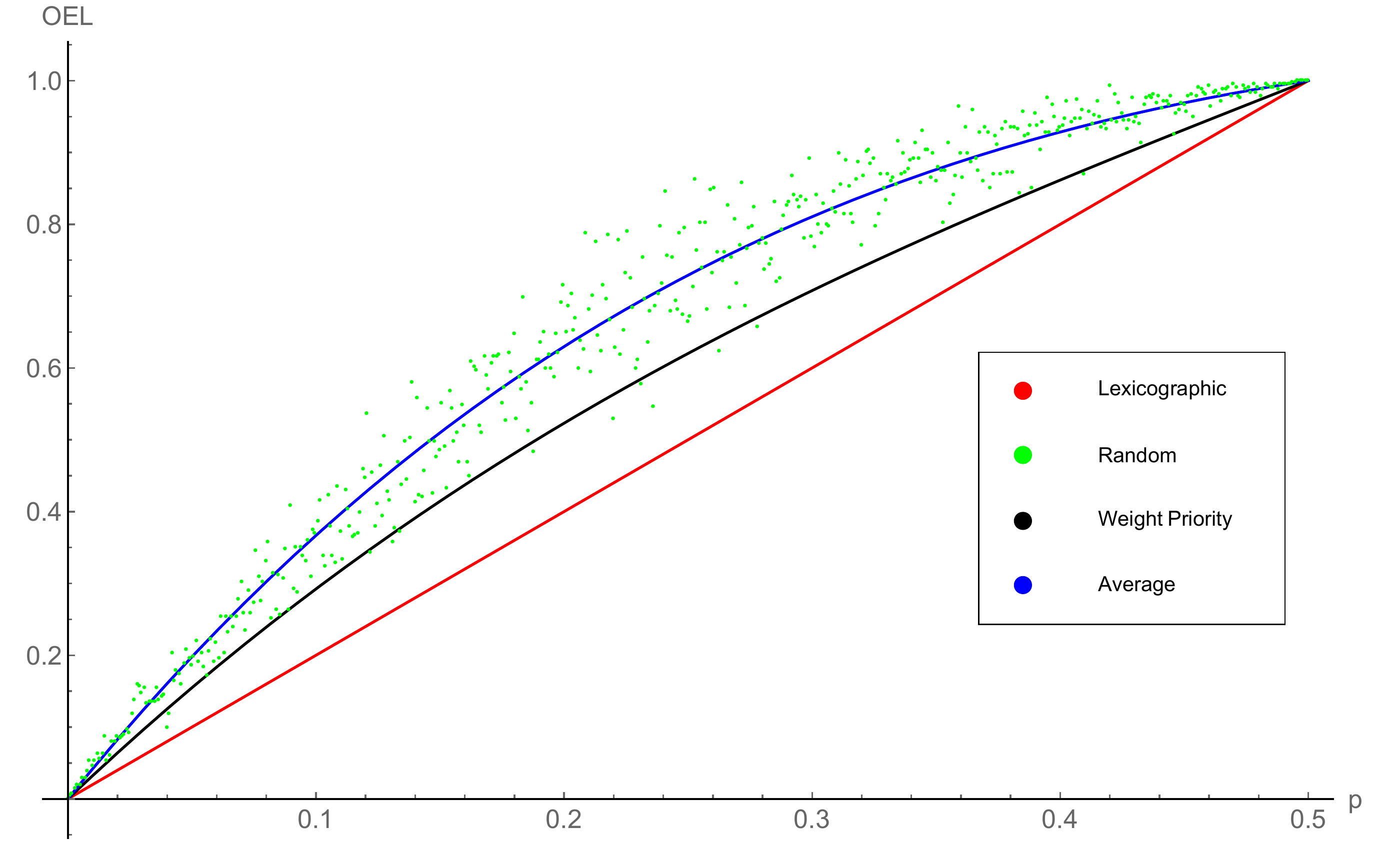}
\caption{$(\mathcal{H}(3),a_{T})$}
\label{fig:HammingT}
\end{minipage}
\end{figure}

\begin{figure}[!h]
\centering
\begin{minipage}[b]{0.47\linewidth}
\includegraphics[width=\linewidth]{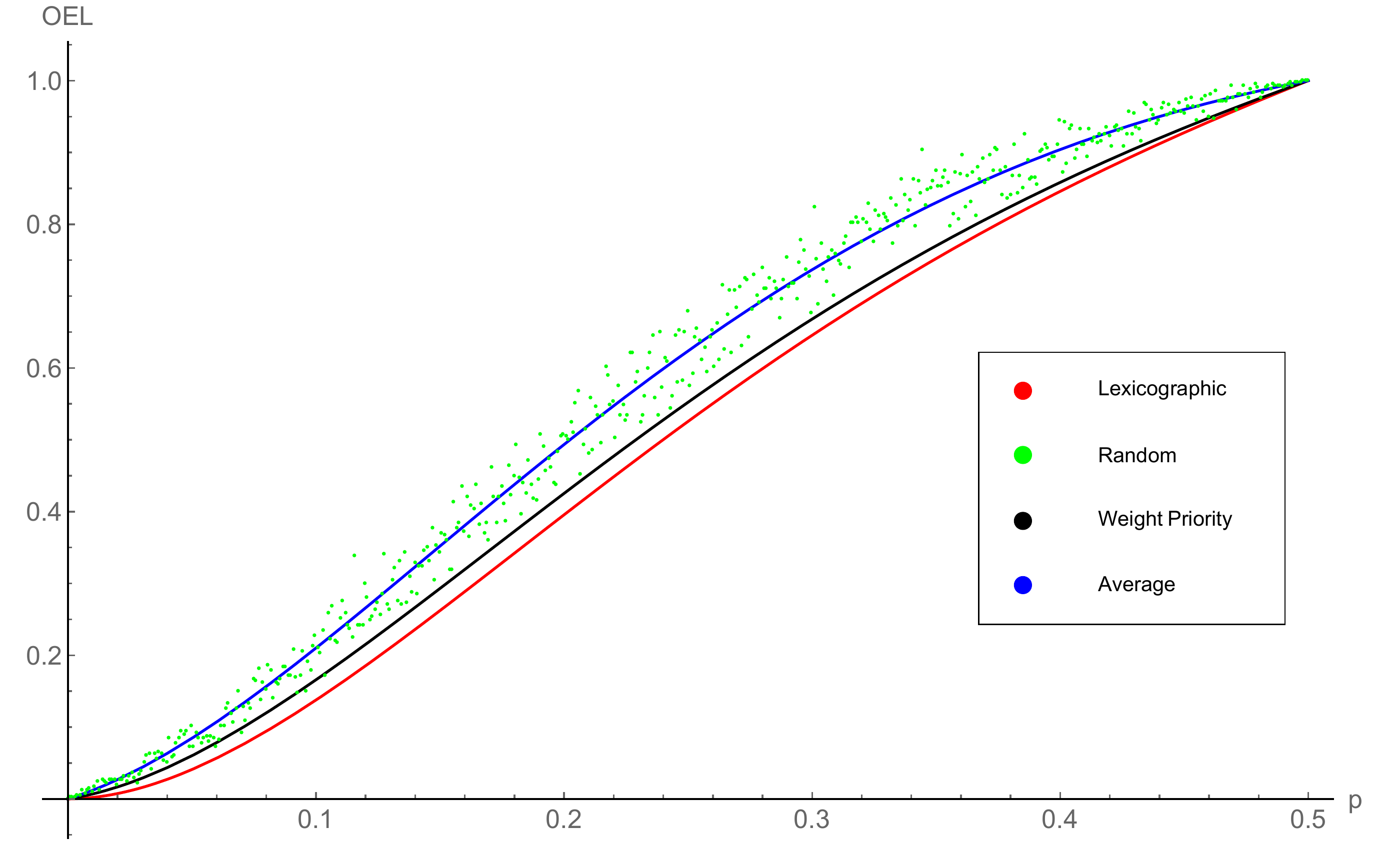}
\caption{$(\mathcal{C}(3),a_{ML})$}
\label{fig:EhammingML}
\end{minipage} \hfill
\begin{minipage}[b]{0.47\linewidth}
\includegraphics[width=\linewidth]{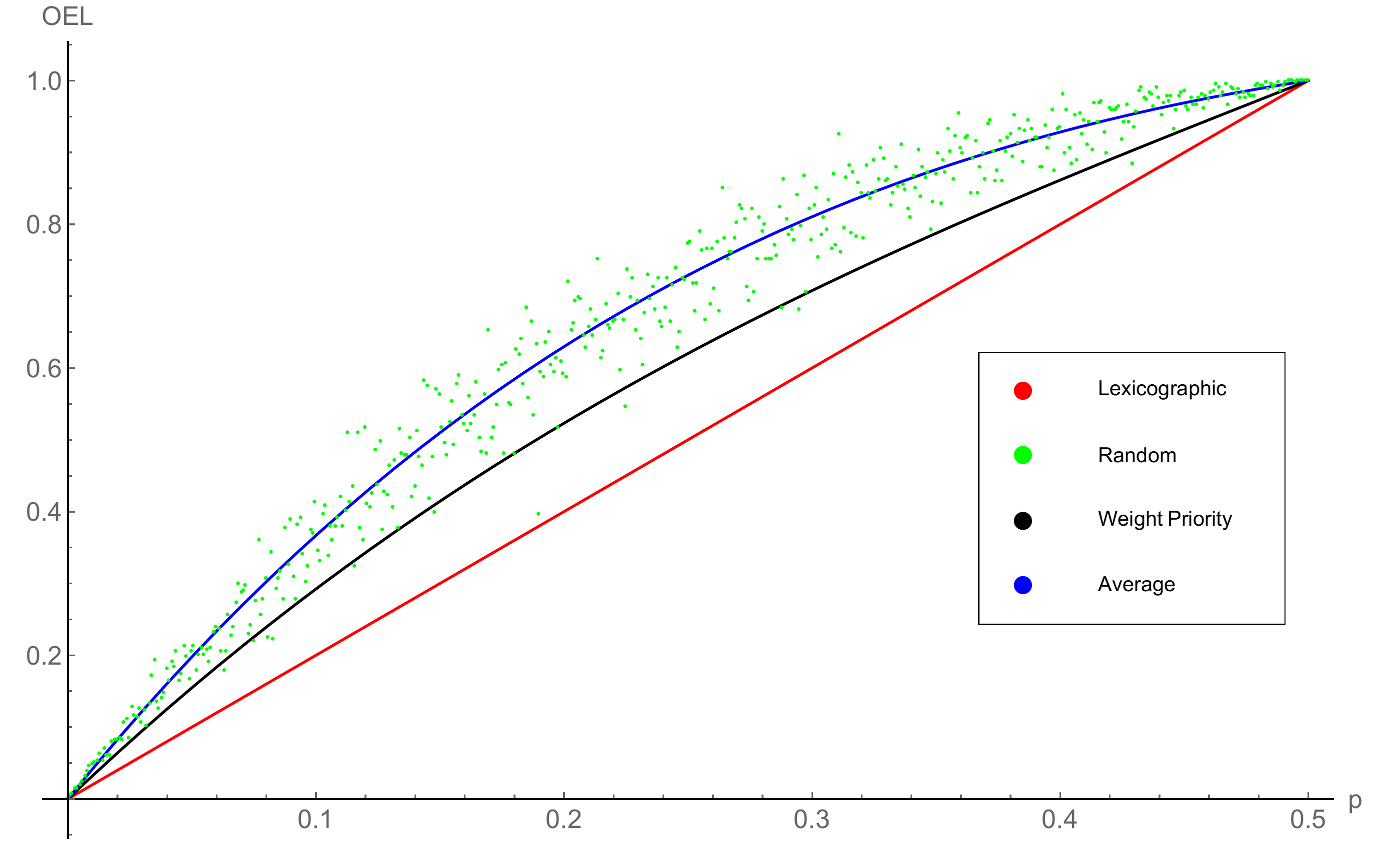}
\caption{$(\mathcal{C}(3),a_{T})$}
\label{fig:EhammingT}
\end{minipage}
\end{figure}

Those simulations support the conjecture that, for those decoders, the
lexicographic encoder is a Bayes encoder.

\subsection{Decoding}

For the encoding part of the problem, we presented the lexicographic encoder as
a candidate to be a Bayes encoder for a pair $\left(  C,a\right)  $, where $C$
is a linear code and $a$ is either the ML-decoder $a_{ML}$ or the decoder
$a_{T}$ defined to be an NN-decoder (according to some metric quite different
from the Hamming one).

For the decoding part of the problem, the situation is more blurry.

The approach adopted in this work is somewhat in the same direction that has
been followed in various recent works regarding unequal error protection (UEP). The proposed use of
 nearest-neighbor decoders determined by a family of ordered
metrics (that will be introduced on the sequence) is actually a decoding
process that gives UEP of bits (bit-wise UEP), in a similar manner to that proposed in
1967 by Masnick and Wolf in \cite{MW} and since then extensively studied by
many authors. Alternately, considering UEP of messages (message-wise
UEP) is the approach adopted by Borade, Nakiboglu and Zheng in \cite{BNZ},
where they consider the necessity of protecting in different ways pieces of information
that are different in their nature (such as data and control messages) or
have different types of errors (erasures and mis-decoded messages). This is
performed by assigning larger decoding regions to the more valuable information.

Our approach is more general, and, in some sense, it combines message-wise and
bit-wise UEP. We protect the messages by placing (encoding) similar (in the
semantic sense) information messages close to each other and by adopting a
decoding criterion that gives priority to the most significant bits.

We consider here two different metrics over $\mathbb{F}_{q}^{n}$, the usual Hamming distance $d_H$ and the total-order metric $d_T\left(x,y\right)=\max\left\{i:x_i\neq y_i\right\}$. Those two metrics are particular instances of the so-called hierarchical poset metrics (see, for example, \cite{PFKH} or \cite{Felix} for an introduction to the subject) and when we do not need to distinguish between them, we may denote the metric as just $d_P$.
As any metric, the
metric $d_{P}$ determines a nearest neighbor (NN) decoder
$a_{P}$: once a message $y$ is received $a_{P}\left(  y\right)  $ is chosen to
minimize the distance to the code, that is, $a_{P}\left(  y\right)  \in
\arg\min\left\{  d_{P}\left(  y,c\right)  :c\in C\right\}  $. In the case of
ambiguity (when $\left\vert \arg\min\left\{  d_{P}\left(  y,c\right)  :c\in
C\right\}  \right\vert \geq2$), we assume the elements in $\arg\min\left\{
d_{P}\left(  y,c\right)  :c\in C\right\}  $ are chosen randomly, with i.i.d. Thus, out of those two metrics, $d_H$ and $d_T$, we consider two different decoders $a_H$ and $a_T$. We remark that this definition of a decoder is actually a list-decoding type definition, and it coincides with the one presented in Section \ref{back} only when $\left\vert \arg\min\left\{  d_{P}\left(  y,c\right)  :c\in
C\right\}  \right\vert =1$. When such an ambiguity exists, by considering an expected loss function, we will  actually  be considering the average (over all of the ambiguities) of the corresponding expected loss functions.

The idea of using hierarchical poset metrics lies in the fact that those
metrics are matched to a lexicographic encoder, in the sense that it gives
more protection against errors in the bits that become more significant due to
the lexicographical manner of encoding. By using a decoder that is not ML, the
number of errors (after decodification) increases, but not uniformly; less
significant errors may increase a greatly, but more significant errors should decrease.

The main question is therefore the following:\textit{ Is there a threshold where the
loss of having more errors is compensated by the reduction in the most significant
errors?} We do not give a conclusive answer to this question but present some
experimental evidence that shows it is indeed the case.

We consider the same toy example as in the encoding part, that is, we consider
each pixel to be attributed a color chosen from a palette with $16$ gray
tones. To explore the decoding side of the problem, we add redundancy by
encoding each pixel as a codeword in the (perfect) $\left[  7;4\right]  _{2}$
binary Hamming code, one codeword assigned for each color that represents a
pixel. Because the lexicographic encoder seems to be the Bayes encoder for
both $a_{ML}$ and $a_{P}$, we consider it as fixed and start to compare
decoders.

Using a random number generator, an error was created for each of the seven
bits of each pixel, with an overall error probability of $p=0.35$. The same received picture
was corrected twice, once using the usual ML decoder and once using an NN-decoder
$a_{T}$, determined by the metric $d_T$.

To illustrate the performance of those decoders, we consider the
same tamarins picture (Figure \ref{mico}) as the original
message. In Figure 10, all of the pixels that were correctly
decoded are painted in purple, while the wrongly decoded pixels
are left as decoded. On the left side, we see the result for the ML
decoder, and on the right side, the result for the NN-decoder $a_T$.

\begin{figure}[!h]
\centering
\begin{minipage}[b]{1\linewidth}
\includegraphics[width=\linewidth]{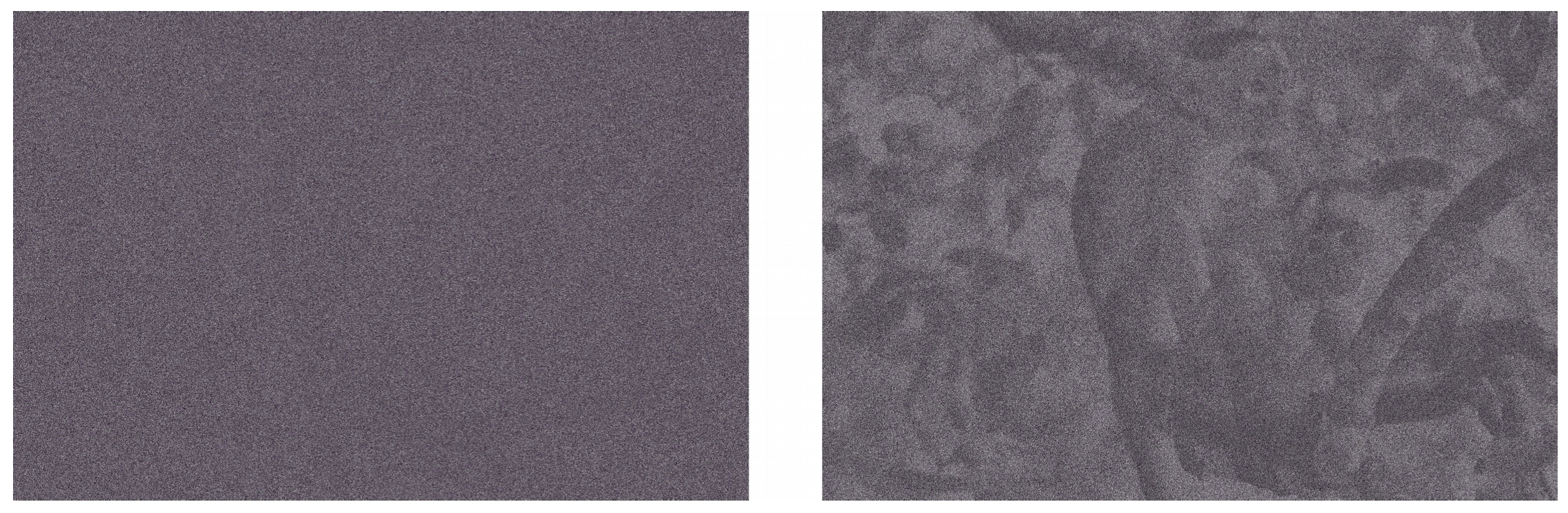}
\caption{Right corrected pixels are colored with
purple.}
\end{minipage}
\end{figure}

As expected, the picture on the left is much more color homogeneous
(purple-like), because using ML to decode with a perfect code minimizes the
amount of errors. However, one can identify the pixels to be painted in purple
only when having the original picture. When looking at the picture as it was
decoded using the two different decoding schemes, one gets a quite different perception:

\begin{figure}[!h]
\centering
\begin{minipage}[b]{1\linewidth}
\includegraphics[width=\linewidth]{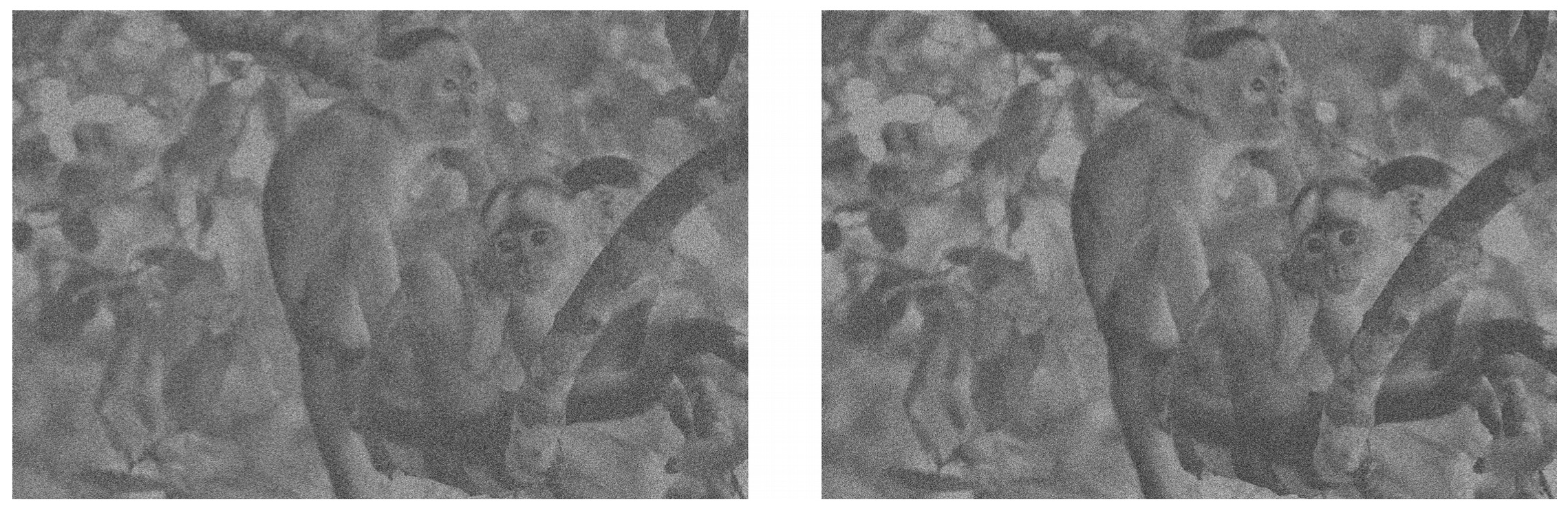}
 \label{fig:9}
   \caption{On the left, ML decoding, and, on the right, NN decoding; $p=0.35$.}
\end{minipage}
\end{figure}

The right-hand image seems to be sharper, closer to the original picture
(Figure 11). This perception regarding the quality of these decoded
pictures is an example of a way of valuing the error, in a situation in
which each of us, ordinary viewers, may be considered as a type of expert.

In Figures 12 and 13 we can see that, even under a very high word error probability ($p=0.4$ and $p=0.43$), it is possible to grasp something of the original message.

\begin{figure}[!h]
\centering
\begin{minipage}[b]{1\linewidth}
\includegraphics[width=\linewidth]{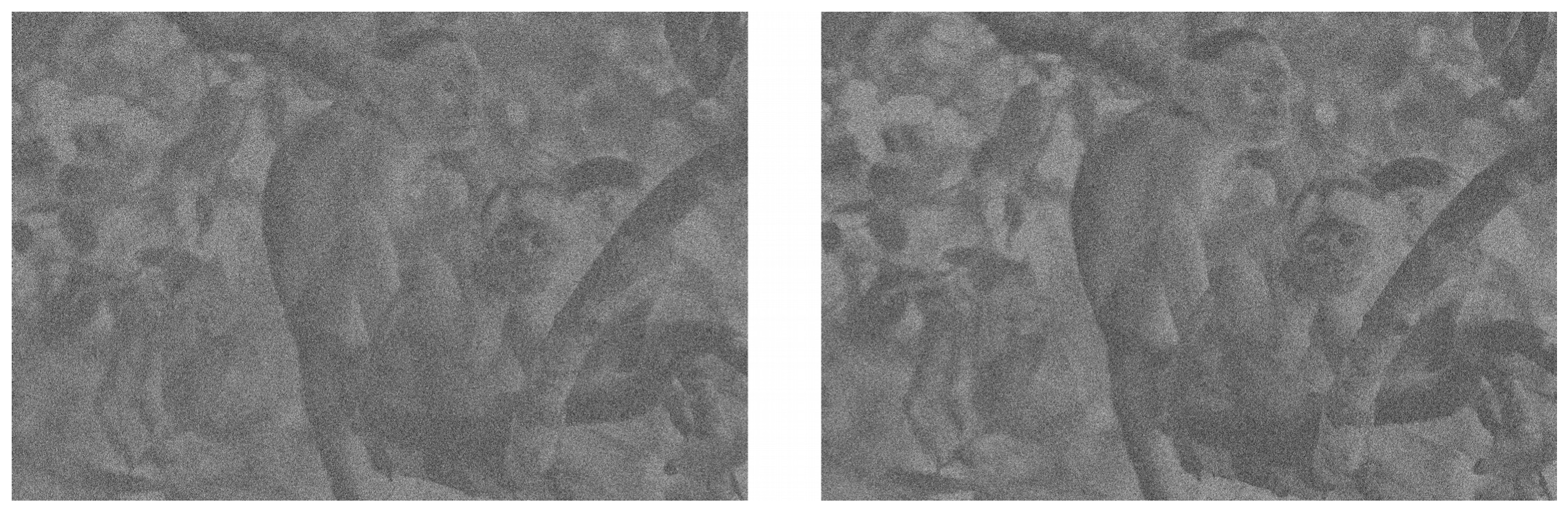}
\label{fig:10}
   \caption{On the left ML decoding and on the right $T$-decoding; $p=0.4$.}
\end{minipage}
\end{figure}

\begin{figure}[!h]
\centering
\begin{minipage}[b]{1\linewidth}
\includegraphics[width=\linewidth]{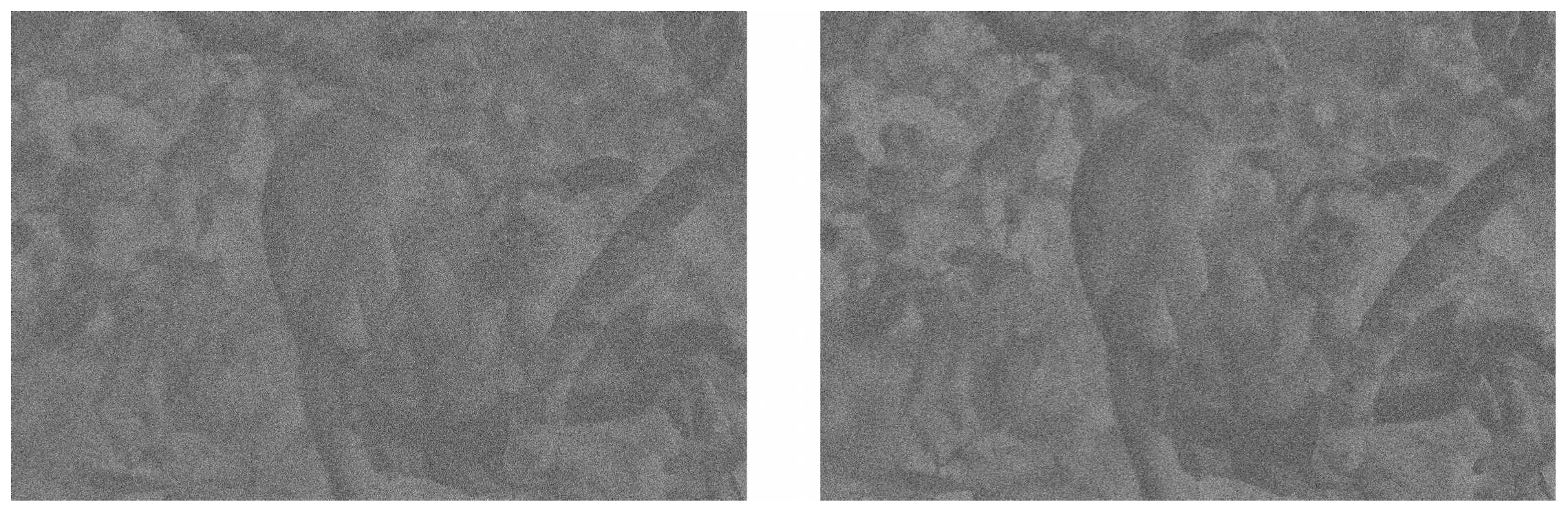}
\label{fig:11}
   \caption{On the left ML decoding and on the right $T$-decoding; $p=0.43$.}
\end{minipage}
\end{figure}

We are able to compute expected loss for those small examples. In the pictures
bellow we graph the expected loss functions in many situations. In each of
them, we are considering one code and  two decoders: the Hamming decoder $a_{H}$ in red and
the total-order decoder $a_{T}$ in black). For each of those codes, we consider
always an lexicographic encoder. The channel is considered to be a BMSC with
overall error probability $p$ and the value function is the one determined by MSE.

\begin{figure}[!h]
\centering
\begin{minipage}[b]{0.47\linewidth}
\includegraphics[width=\linewidth]{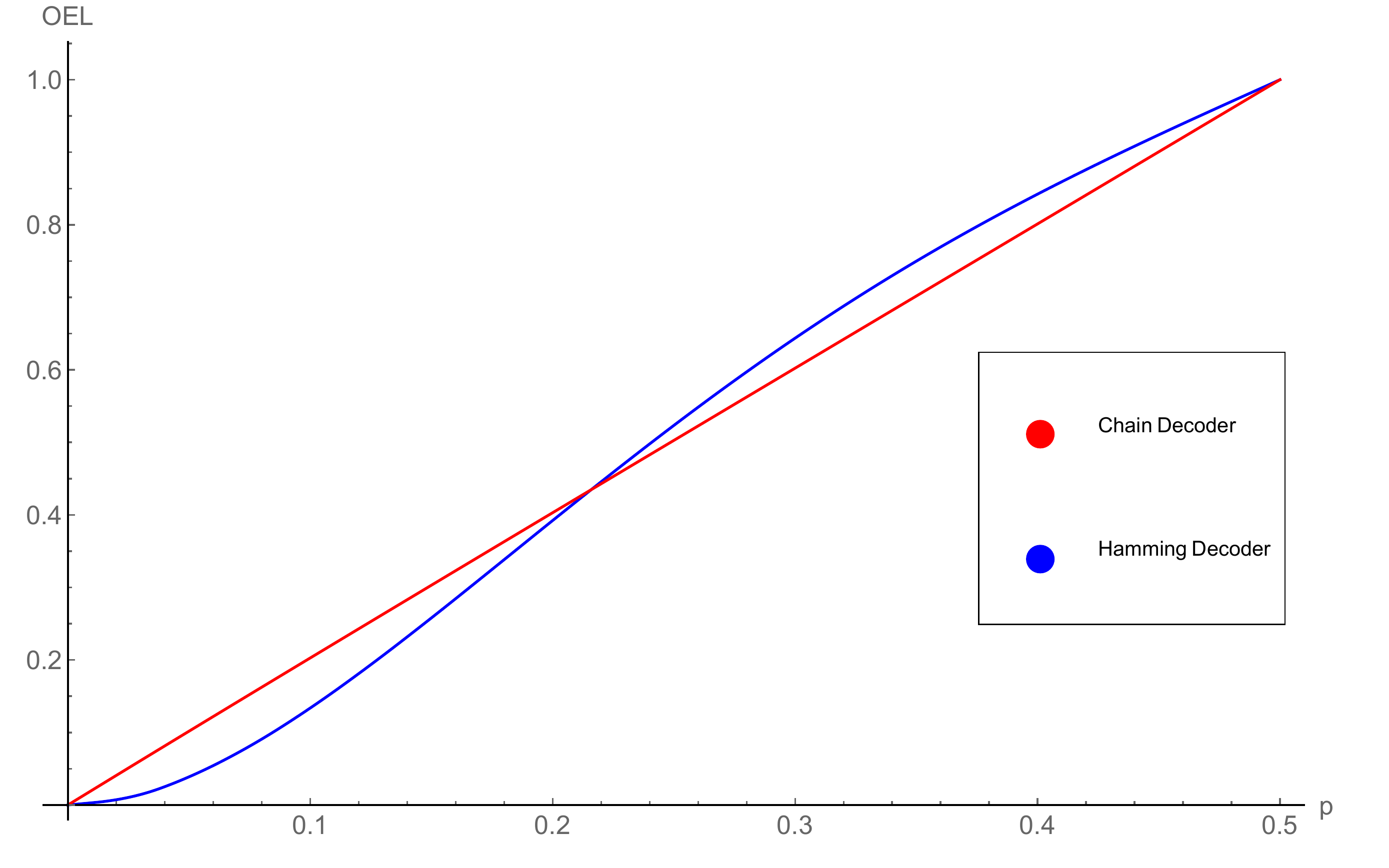}
\caption{Code $\mathcal{H}(3)$.}
\label{fig:codeH(3)}
\end{minipage} \hfill
\begin{minipage}[b]{0.47\linewidth}
\includegraphics[width=\linewidth]{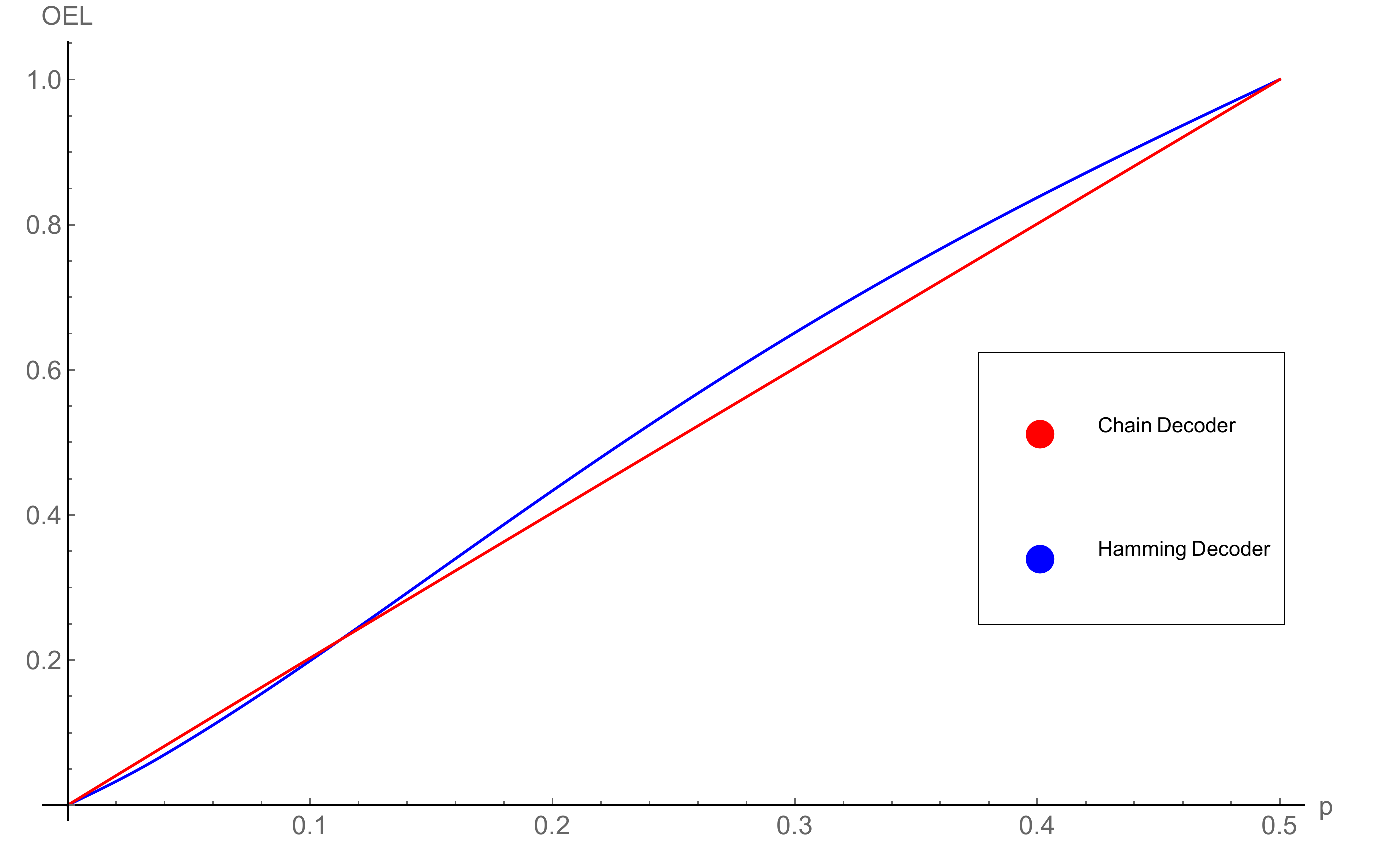}
\caption{Code $\mathcal{C}(3)$.}
\label{fig:codeC(3)}
\end{minipage}
\end{figure}

\begin{figure}[!h]
\centering
\begin{minipage}[b]{0.47\linewidth}
\includegraphics[width=\linewidth]{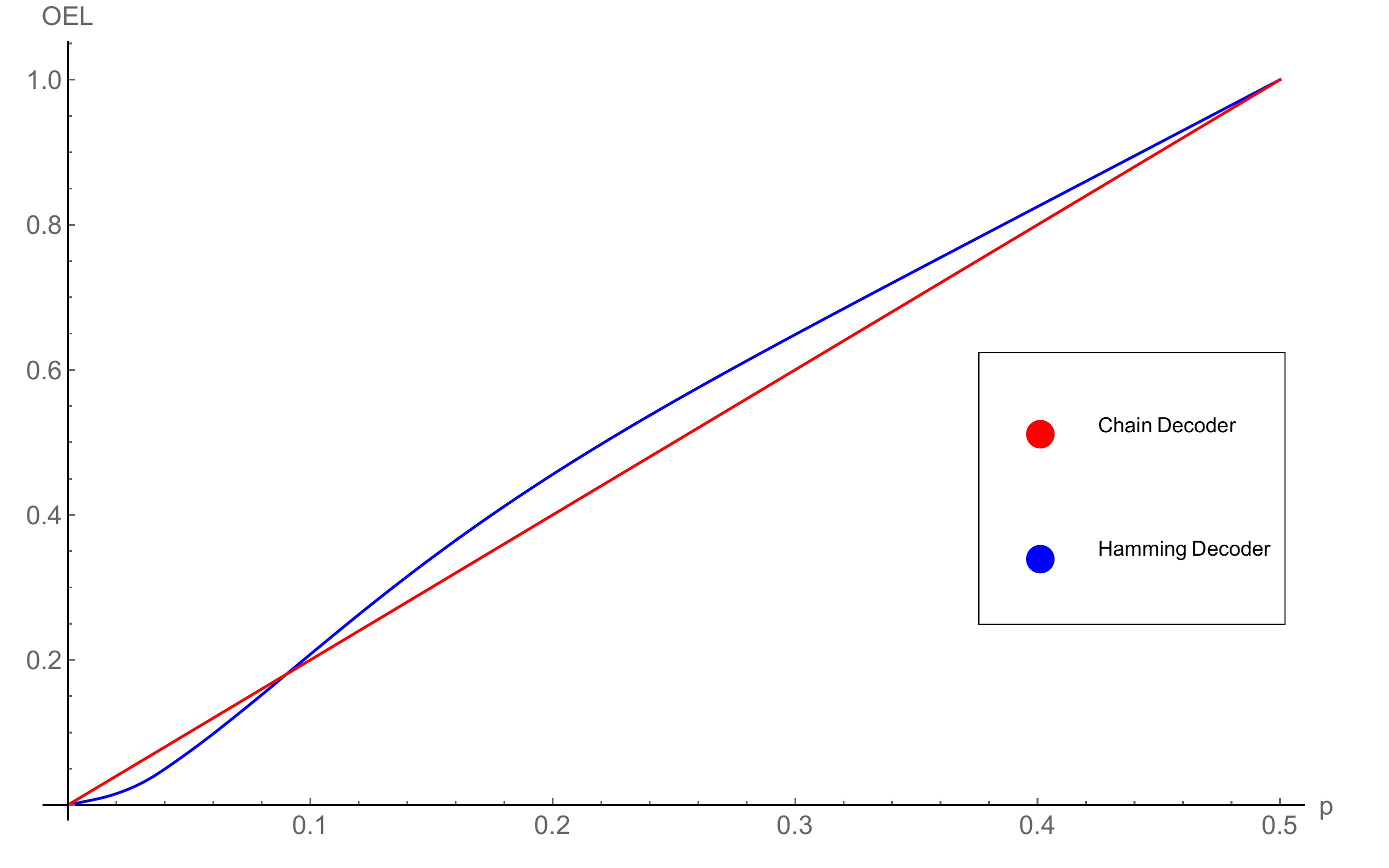}
\caption{Code $\mathcal{H}(4)$.}
\label{fig:codeH(4)}
\end{minipage} \hfill
\begin{minipage}[b]{0.47\linewidth}
\includegraphics[width=\linewidth]{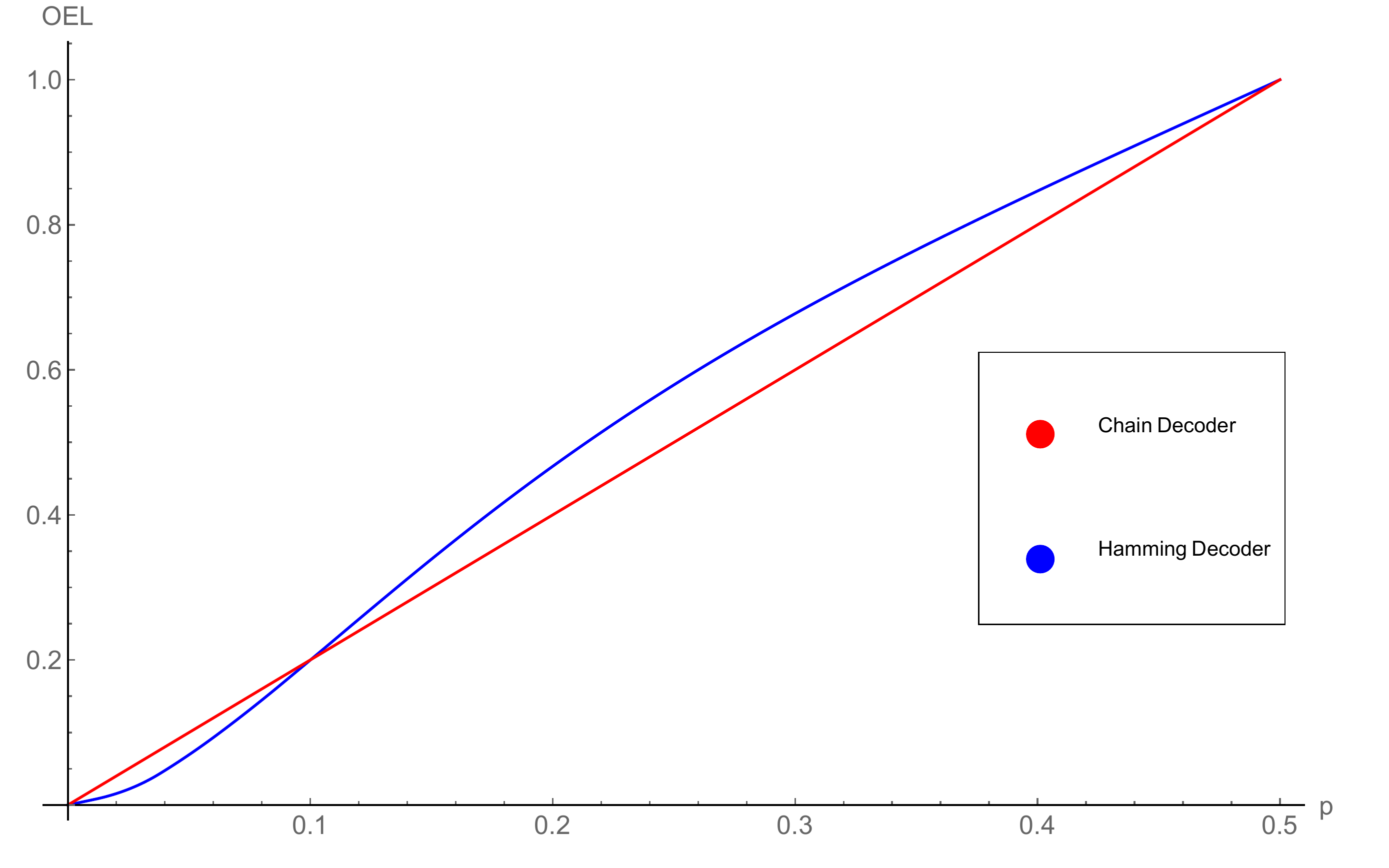}
\caption{Code $\mathcal{C}(4)$.}
\label{fig:codeC(4)}
\end{minipage}
\end{figure}

The first two pictures represent the situation with $4$ bits of color
(information) and the same redundancy, so we consider two $\left[  7;4\right]
_{2}$ linear codes: the Hamming code $\mathcal{H}\left(  3\right)  $ and the
code $\mathcal{C}\left(  3\right)  $, both introduced in the end of Section
\ref{encoding}. In the last two pictures, we consider $11$ bits of color
(information) and the same redundancy, so we consider two $\left[
15;11\right]  _{2}$ linear codes: the Hamming code $\mathcal{H}\left(
4\right)  $ and the code $\mathcal{C}\left(  4\right)  $ that has a parity
check matrix
\[
\left(
\begin{array}
[c]{ccccccccccccccc}%
1 & 0 & 0 & 0 & 1 & 1 & 1 & 0 & 0 & 0 & 0 & 0 & 1 & 1 & 0\\
0 & 1 & 0 & 0 & 0 & 1 & 0 & 1 & 1 & 1 & 0 & 1 & 0 & 1 & 0\\
0 & 0 & 1 & 0 & 0 & 0 & 0 & 1 & 0 & 0 & 0 & 1 & 0 & 0 & 1\\
0 & 0 & 0 & 1 & 0 & 0 & 0 & 0 & 0 & 1 & 0 & 1 & 0 & 1 & 1
\end{array}
\right)  \text{.}%
\]

As we can see, for small $p$, the ML decoder seems to be the best choice.
Nevertheless, those pictures depict the existence of a threshold $p_{0}$ above
which the ordered decoders should be used. We do not know what would be
the behavior of such a threshold for large $n$, but it seems that, for (very)
noisy channels, using those ordered decoders may improve the performance
and yield an extra bonus: very efficient decoding algorithms (details in
\cite{Felix}).

\section{Error Value Functions that are Invariant by Translations\label{translation}}

In this section, we assume that the information set $\mathcal{I}$ is identified
with the vector field $\mathbb{F}_{q}^{k}$, that is, we assume $\mathcal{I}%
=\mathbb{F}_{q}^{k}$, and we consider a special class of error value functions over
$\mathbb{F}_{q}^{k}$, those that are invariant by translations. We assume also
that $\mathcal{X}=\mathcal{Y}=\mathbb{F}_{q}^{n}$ and that  $C$ is an $\left[
n;k\right]  _{q}$ linear code.

We say that an error value function $\nu:\mathbb{F}_{q}^{k}\times
\mathbb{F}_{q}^{k}\rightarrow\mathbb{R}_{+}$ is \textit{invariant by
translations} if
\begin{equation}
\nu\left(  x+z,y+z\right)  =\nu\left(  x,y\right)  \label{invariant}%
\end{equation}
for all $x,y,z\in\mathbb{F}_{q}^{k}$. In this case, we may consider a function
$\widetilde{\nu}:\mathbb{F}_{q}^{k}\rightarrow\mathbb{R}_{+}$ (also called an
\textit{error value function}) defined by $\widetilde{\nu}\left(  x\right)
:=\nu\left(  x,0\right)  $ because $\nu\left(  x,y\right)  =\nu\left(
x-y,0\right)  $.

We remark that, whenever there is a significant \emph{real number model} for
the information as an injective map $m:\mathcal{I}\rightarrow\mathbb{R}$, an
error value function $\nu$ is said to be \emph{compatible with the model } if
$\mu\left(  \iota_{1},\iota_{2}\right)  $ is increasing with the difference
$\left\vert m\left(  \iota_{1}\right)  -m\left(  \iota_{2}\right)  \right\vert
$. It is not difficult to see that an error value function invariant by translations
can not be compatible with a real number model for the information, and, in such
a situation, a translation invariant error value function can be regarded only as a
simplified model for errors. We also remark that, when defining invariance by
translations, we assumed that the information set $\mathcal{I}$ is identified
with the vector space $\mathbb{F}_{q}^{k}$, let us say by a bijection
$\sigma:\mathcal{I}\rightarrow\mathbb{F}_{q}^{k}$. The identity
(\ref{invariant}) actually depends on $\sigma$, so that in fact we should say
$\nu$ is $\sigma$-invariant invariant by translations. Nevertheless, we assume
$\sigma$ is given and fixed, and hence the notation we adopt ignores its role.

The advantage of error value functions that are invariant by translations is that
this allows us to determine Bayes encoders, as we see in the next two propositions.

As stated in equation (\ref{perda esperada}), we express
\[
\mathbb{E}\left(  C,a,\nu_{f}\right)  =\sum_{\left(  c,c^{\prime}\right)  \in
C\times C}G_{a}\left(  c,c^{\prime}\right)  \nu_{f}\left(  c,c^{\prime
}\right)
\]
where $G_{a}\left(  c,c^{\prime}\right)  =\sum_{y\in a^{-1}\left(  c\right)
}P\left(  \left.  c^{\prime}\right\vert y\right)  P\left(  y\right)  $.
Assuming now that the encoder $f:\mathbb{F}_{q}^{k}\rightarrow C$ is a linear
map, we have that $\nu_{f}\left(  c,c^{\prime}\right)  =\nu_{f}\left(
c-c^{\prime},0\right)  =\widetilde{\nu}_{f}\left(  c-c^{\prime}\right)  $, and,
writing $u=c-c^{\prime}$, we find that
\begin{align*}
\mathbb{E}\left(  C,a,\nu_{f}\right)   &  =\sum_{\left(  c,c^{\prime}\right)
\in C\times C}G_{a}\left(  c,c^{\prime}\right)  \widetilde{\nu}_{f}\left(
c-c^{\prime}\right) \\
&  =\sum_{u\in C}\left(  \sum_{c-c^{\prime}=u}G_{a}\left(  c,c^{\prime
}\right)  \right)  \widetilde{\nu}_{f}\left(  u\right) \\
&  =\sum_{u\in C}\left(  \sum_{c\in C}G_{a}\left(  c,c-u\right)  \right)
\widetilde{\nu}_{f}\left(  u\right)
\end{align*}
and we get the following:

\begin{proposition}
If $\nu$ is an error value function invariant by translations and $f$ is a
linear encoder, then
\[
\mathbb{E}\left(  C,a,\nu_{f}\right)  =\sum_{c\in C}H_{a}(c)\widetilde{\nu
}_{f}\left(  c\right)
\]
where
\[
H_{a}(c):=\sum_{c\in C}G_{a}\left(  c,c-u\right)  =\sum_{y\in\mathbb{F}%
_{q}^{n}}P\left(  a(y)-c|y\right)  P(y).
\]

\end{proposition}

We remark that $\mathbb{E}\left(  C,a,\nu_{f}\right)  $ is a linear
combination of the values $\widetilde{\nu}_{f}\left(  c\right)  $, $c\in C$,
and that the coefficients $H_{a}(c)$ do not depend on the encoder $f$. Because the
values $\widetilde{\nu}_{f}\left(  c\right)  $ are all non-negative, it is
simple to minimize the expected loss; one should choose an encoder $f$ that
associates more valuable errors (higher $\widetilde{\nu}_{f}\left(
c\right)  $) to the lower coefficients $H_{a}(c)$.

\begin{theorem}
\label{teo. cod. bayes} Let $C=\left\{  c_{1},\ldots,c_{M}\right\}  $ be a
linear code, $a:\mathbb{F}_{q}^{n}\rightarrow C$ a decoder and $\nu
:\mathbb{F}_{q}^{k}\times\mathbb{F}_{q}^{k}\rightarrow\mathbb{R}_{+}$ an error
value function invariant by translations. We assume without loss of generality
that
\[
H_{a}\left(  c_{1}\right)  \geq H_{a}\left(  c_{2}\right)  \geq\ldots\geq
H_{a}\left(  c_{M}\right)  .
\]
Then, a linear encoder $f:\mathbb{F}_{q}^{k}\rightarrow C$ is a Bayes encoder
for the error value function $\nu$ and the decoder $a$ if and only if
\[
\widetilde{\nu}_{f}\left(  c_{1}\right)  \leq\widetilde{\nu}_{f}\left(
c_{2}\right)  \leq\cdots\leq\widetilde{\nu}_{f}\left(  c_{M}\right)  \text{.}%
\]

\end{theorem}

From here on, we assume that the prior probability of $C$ is uniform.



We now explore an example that illustrates the preceding results for the case of a perfect code. Let us assume
the channel to be a binary memoryless symmetric channel with overall error
probability $p$. Let
$w_{H}(c):=d_{H}\left(  c,0\right)  $ be the Hamming weight of $c$ and let
$supp(x):=\{i:x_{i}\neq0\}$ be the \textit{support} of the vector
$x=(x_{1},\ldots,x_{n})$.

Let $\mathcal{H}(l)$ be the  $\left[  n;k\right]  _{2}$\emph{ binary Hamming
code}, where $n=2^{l}-1$ is the length of the block and $k=2^{l}-1-l$ the
dimension of the code. Let $a:\mathbb{F}_{2}^{n}\rightarrow\mathcal{H}(l)$ be
an ML decoder, so that $d_{H}\left(  y,a\left(  y\right)  \right)  \leq
d_{H}\left(  y,c\right)  $, for all $y\in\mathbb{F}_{2}^{n}$ and every
$c\in\mathcal{H}\left(  l\right)  $. In this situation, direct computations show that%
\begin{equation}
H_{a}\left(  c\right)  =\frac{\left(  1-p\right)  ^{n}}{M}\sum_{y\in
\mathbb{F}_{2}^{n}}s^{d_{H}\left(  y,a\left(  y\right)  -c\right)  }
\label{H_a}%
\end{equation}
where $s:=\frac{p}{1-p}$ and $M:=\left\vert \mathcal{H}(l)\right\vert $. Since
$\mathcal{H}\left(  l\right)  $ is a perfect code that corrects a single error
and $a$ is an ML decoder, we find that either $y=a\left(  y\right)  $ or
$y-a\left(  y\right)  =e_{i}$, a vector with Hamming weight $w_{H}\left(
e_{i}\right)  =1$. This ensures that
\[
d_{H}\left(  y,a\left(  y\right)  -c\right)  =
\left\{
\begin{array}
[c]{cc}%
w_{H}\left(  e_{i}+c\right) & \text{for some }1\leq i\leq N\text{ if }%
y\notin\mathcal{H}\left(  l\right) \\
w_{H}\left(  c\right) & \text{ otherwise}%
\end{array}
\right.  \text{.}%
\]
We write
\begin{align*}
\sum_{y\in\mathbb{F}_{2}^{n}}s^{d_{H}\left(  y,a\left(  y\right)  -c\right)} & =
\sum_{y\in\mathcal{H}\left(  l\right)  }s^{d_{H}\left(  y,a\left(
y\right)  -c\right)  }+\sum_{y\notin\mathcal{H}\left(  l\right)  }%
s^{d_{H}\left(  y,a\left(  y\right)  -c\right)  } \\ &
=\sum_{y\in\mathcal{H}\left(  l\right)  }s^{w_{H}\left(  c\right)  }%
+M\cdot\sum_{i\notin\text{\emph{supp}}\left(  c\right)  }s^{w_{H}\left(
e_{i}+c\right)  }
+M\cdot\sum_{i\in\text{\emph{supp}}\left(  c\right)
}s^{w_{H}\left(  e_{i}+c\right)  }\text{.}
\end{align*}
Because
\[
\sum_{i\notin\text{\emph{supp}}\left(  c\right)  }s^{w_{H}\left(
e_{i}+c\right)  }=\left(  n-w_{H}\left(  c\right)  \right)  s^{w_{H}\left(
c\right)  +1}%
\]
and
\[
\sum_{i\in\text{\emph{supp}}\left(  c\right)  }s^{w_{H}\left(  e_{i}+c\right)
}=w_{H}\left(  c\right)  s^{w_{H}\left(  c\right)  -1}\text{,}%
\]
considering expression (\ref{H_a}), it follows that
\[
H_{a}\left(  c\right)  =\left(  1-p\right)  ^{n}  (s^{w_{H}\left(
c\right)  }+\left(  n-w_{H}\left(  c\right)  \right)  s^{w_{H}\left(
c\right)  +1}
+w_{H}\left(  c\right)  s^{w_{H}\left(  c\right)  -1})
\text{,}%
\]
which depends only on $w_{H}\left(  c\right)  $ but not on $c$. Now it is
possible to show that, for the Hamming code $\mathcal{H}\left(  l\right)  $,
$H_{a}\left(  c\right)  \geq H_{a}\left(  c^{\prime}\right)  $ if
$w_{H}\left(  c\right)  \leq w_{H}\left(  c^{\prime}\right)  $. Using this and
Theorem \ref{teo. cod. bayes}, it follows that $f:\mathbb{F}_{2}%
^{k}\rightarrow\mathcal{H}\left(  l\right)  $ is a Bayes encoder of a
Hamming code iff
\begin{equation}
\widetilde{\nu}_{f}\left(  c\right)  \leq\widetilde{\nu}_{f}\left(  c^{\prime
}\right)  \text{ whenever }w_{H}\left(  c\right)  \leq w_{H}\left(  c^{\prime
}\right)  \text{.} \label{perfect}%
\end{equation}

Similar reasoning may be used to compute the coefficients $H_{a}\left(
c\right)  $ of the perfect $[23;12]_{2}$ Golay code and to show that condition
(\ref{perfect}) holds also for this code. We call such an encoder a \textit{weight
priority encoder}.

Finally, in addition to giving a good description of Bayes encoders, the use of a value function that is invariant by translation may be justified also by the fact that it generalizes two well-known measures of loss. When introducing the expected loss in Section
\ref{section 1}, we already showed that the word error probability may be
considered as a particular instance of an expected loss function (by
considering the indicator function $\nu_{0\text{-}1}$ to be the loss
function). If we assume now that the channel is \textit{invariant by
translations}, in the sense that $$P\left(  \left.  y\right\vert x\right)
=P\left(  \left.  y+z\right\vert x+z\right)  $$ for all $x,y,z\in\mathbb{F}%
_{q}^{n}$, we may look also at the \textit{bit error probability} $P_{b}(C)$
of $C$ (BER) as a particular case of the expected loss function. This is
attained by considering a decoder $a:\mathbb{F}_{q}^{n}\rightarrow C$ that is
also \textit{invariant by translations}, in the sense that $$a^{-1}\left(  c\right)
=c+a^{-1}\left(  0\right)  $$ for every $c\in C$, considering any encoder
$f:\mathbb{F}_{q}^{k}\rightarrow C$ and the value function $\nu:\mathbb{F}%
_{q}^{k}\rightarrow\mathbb{R}_{+}$ defined by
\[
\nu_{f}\left(  c\right)  =\frac{w_{H}\left(  f^{-1}\left(  c\right)  \right)
}{k}\text{.}%
\]
The proof follows by direct calculations and comparison with the expression for
BER, as given, for example, in \cite{Bigliere}.

\section*{Acknowledgment}

M. Firer wish to acknowledge the support of São Paulo Research Foundation
(FAPESP) through grants 2013/25977-7 and 2013/09493-0. The illustrating pictures
 in this work were produced using a software developed by
Vanderson Martins do Rosario, an undergraduate student at Universidade Estadual
de Maringá (UEM), while he was in high school. The authors are deeply grateful
to Vanderson. This work was presented in part at ITA 2013.

\end{document}